\documentclass[11pt,dvipsnames]{article}
\usepackage[margin=1in]{geometry}
\usepackage{amsmath,amsthm,amsfonts,amssymb,amscd,bm,mathrsfs,extarrows}
\usepackage{fancyhdr}
\usepackage{xcolor}
\usepackage{graphicx}
\usepackage[colorlinks=true, urlcolor=blue, linkcolor=blue, citecolor=magenta]{hyperref}
\usepackage{multirow}
\usepackage{comment}
\usepackage{mathpazo}
\usepackage{enumitem}
\usepackage[T1]{fontenc}

\allowdisplaybreaks

\usepackage[capitalise,nameinlink]{cleveref}
\Crefname{lemma}{Lemma}{Lemmas}
\Crefname{fact}{Fact}{Facts}
\Crefname{theorem}{Theorem}{Theorems}
\Crefname{corollary}{Corollary}{Corollaries}
\Crefname{claim}{Claim}{Claims}
\Crefname{example}{Example}{Examples}
\Crefname{problem}{Problem}{Problems}
\Crefname{definition}{Definition}{Definitions}
\Crefname{notation}{Notation}{Notations}
\Crefname{assumption}{Assumption}{Assumptions}
\Crefname{subsection}{Subsection}{Subsections}
\Crefname{section}{Section}{Sections}
\Crefformat{equation}{(#2#1#3)}

\newtheorem{theorem}{Theorem}[section]
\newtheorem*{theorem*}{Theorem}

\newtheorem{proposition}[theorem]{Proposition}
\newtheorem*{proposition*}{Proposition}
\newtheorem{lemma}[theorem]{Lemma}
\newtheorem*{lemma*}{Lemma}
\newtheorem{corollary}[theorem]{Corollary}
\newtheorem*{corollary*}{Corollary}
\newtheorem*{conjecture*}{Conjecture}
\newtheorem{fact}[theorem]{Fact}
\newtheorem*{fact*}{Fact}

\newtheorem*{exercise*}{Exercise}
\newtheorem{hypothesis}[theorem]{Hypothesis}
\newtheorem*{hypothesis*}{Hypothesis}

\theoremstyle{definition}
\newtheorem{definition}[theorem]{Definition}

\newtheorem{exercise-easy}[theorem]{Exercise}
\newtheorem{exercise-med}[theorem]{Exercise}
\newtheorem{exercise-hard}[theorem]{Exercise$^\star$}
\newtheorem{claim}[theorem]{Claim}
\newtheorem*{claim*}{Claim}

\newtheorem*{remark*}{Remark}

\newtheorem*{observation*}{Observation}

\renewcommand{\Pr}{\operatorname*{\mathbf{Pr}}}

\newcommand{\val}{\mathsf{val}}

\newcommand{\gapkcsp}{\textsc{$\varepsilon$-Gap $k$-CSP}}

\newcommand{\eps}{\varepsilon}

\newcommand{\abra}[1]{\left\langle #1 \right\rangle}
\newcommand{\pbra}[1]{\left( #1 \right)}
\newcommand{\sbra}[1]{\left[ #1 \right]}
\newcommand{\cbra}[1]{\left\{ #1 \right\}}

\newcommand{\ceilbra}[1]{\left\lceil #1 \right\rceil}
\newcommand{\bin}{\{0,1\}}

\newcommand{\poly}{\mathsf{poly}}

\newcommand{\indicator}{\mathbb{1}}

\newcommand{\Fbb}{\mathbb{F}}

\newcommand{\Ccal}{\mathcal{C}}

\newcommand{\Ecal}{\mathcal{E}}

\newcommand{\Scal}{\mathcal{S}}

\newcommand{\Vcal}{\mathcal{V}}

\newcommand{\pwh}{\mathtt{PWH}}
\newcommand{\wh}{\mathtt{WH}}

\newcommand{\vcsp}{VecCSP}

\renewcommand{\tilde}{\widetilde}
\renewcommand{\bar}{\overline}
\renewcommand{\hat}{\widehat}

\title{
Parameterized Inapproximability Hypothesis under ETH
}

\author{
Venkatesan Guruswami\thanks{Simons Institute for the Theory of Computing, and Departments of EECS and Mathematics, UC Berkeley. Email: {\tt venkatg@berkeley.edu}. Research supported in part by NSF grants CCF-2228287 and CCF-2211972 and a Simons Investigator award. }
\and 
Bingkai Lin\thanks{State Key Laboratory for Novel Software Technology, Nanjing University. Email: \texttt{lin@nju.edu.cn}}
\and 
Xuandi Ren\thanks{Department of EECS, UC Berkeley. Email: \texttt{xuandi\_ren@berkeley.edu}. Supported in part by NSF grant CCF-2228287.} 
\and
Yican Sun\thanks{School of Computer Science, Peking University. Email: \texttt{sycpku@pku.edu.cn}}
\and
Kewen Wu\thanks{Department of EECS, UC Berkeley. Email: \texttt{shlw\_kevin@hotmail.com}. Supported by a Sloan Research Fellowship and NSF CAREER Award CCF-2145474.}
}
\date{}

\begin{document}

\maketitle
\thispagestyle{empty}

\begin{abstract}
The Parameterized Inapproximability Hypothesis (PIH) asserts that no fixed parameter tractable (FPT) algorithm can distinguish a satisfiable CSP instance, parameterized by the number of variables, from one where every assignment fails to satisfy an $\varepsilon$ fraction of constraints for some absolute constant $\varepsilon > 0$. PIH plays the role of the PCP theorem in parameterized complexity. However, PIH has only been established under Gap-ETH, a very strong assumption with an inherent gap.

\smallskip
In this work, we prove PIH under the Exponential Time Hypothesis (ETH). This is the first proof of PIH from a gap-free assumption. Our proof is self-contained and elementary. We identify an ETH-hard CSP whose variables take vector values, and constraints are either linear or of a special parallel structure. Both kinds of constraints can be checked with constant soundness via a ``parallel PCP of proximity'' based on the Walsh-Hadamard code. 
\end{abstract}

\newpage
\tableofcontents
\thispagestyle{empty}

\clearpage

\section{Introduction}\label{sec:intro}

A comprehensive understanding of {\sf NP}-hard problems is an everlasting pursuit in the TCS community. 
Towards this goal, researchers have proposed many alternative hypotheses as strengthenings of the classic ${\sf P\ne NP}$ assumption to obtain more fine-grained lower bounds for {\sf NP}-hard problems, for example, Exponential Time Hypothesis (ETH)~\cite{IP01}, Strong Exponential Time Hypothesis (SETH)~\cite{IP01, cip09}, Gap Exponential Time Hypothesis (Gap-ETH)~\cite{Din16}.

Besides a richer family of hypotheses, \emph{approximation} and \emph{fixed parameter tractability (FPT)} are also two  orthogonal approaches to cope with {\sf NP}-hardness.
\begin{itemize}
\itemsep=0ex
    \item In the approximation setting, we consider optimization problem, where  input instances are associated with a cost function and the goal is to find a solution with cost function value close to the optimum.
    \item In the fixed parameter tractability (FPT) setting, every instance is attached with a parameter $k$ indicating specific quantities (e.g., the optimum, the treewidth) of the instance. This setting treats $k$ as a parameter much smaller than the instance size $n$, i.e., $1\le k\ll n$. Thus, the required runtime of the algorithm is relaxed from $n^{O(1)}$ to $f(k)\cdot n^{O(1)}$, for any computable function $f$. The class {\sf FPT} is the set of parameterized problems that admit an algorithm within this running time.
    
    The seminal studies in this setting built up \emph{parameterized complexity theory}~\cite{DowneyF95,downey1995fixed,flugro06}. In this theory, there are also alternative hypotheses as a strengthening of ${\sf P\ne NP}$.
    For example, ${\sf W[1]\ne FPT}$, which is equivalent to the statement that {\sc $k$-Clique} has no $f(k)\cdot n^{O(1)}$-time algorithm.
\end{itemize}

Recently, there has been an extensive study at the intersection of these two settings: the existence (or absence) of approximation algorithms that solve {\sf NP}-hard problems in FPT time. 
\begin{itemize}
\itemsep=0ex
    \item On the algorithmic side, approximation algorithms with FPT runtime have been designed for various {\sf NP}-complete problems. Examples include {\sc Vertex-Coloring}~\cite{dhk05,mar08}, {\sc Min-$k$-Cut} \cite{gupta2018fpt,gupta2018faster,kawarabayashi2020nearly,lokshtanov2022parameterized}, {\sc $k$-Path-Deletion}~\cite{lee19}, {\sc $k$-Clustering} \cite{abbasi2023parameterized}, {\sc $k$-Means} / {\sc $k$-Medians} \cite{cgts02,KMN+04,LiS16,BPR+17,cgk+19,answ20}, {\sc Max $k$-Hypergraph Vertex Cover}~\cite{sf17,man19}, {\sc Flow Time Scheduling}~\cite{Wie18}.
    
    \item In terms of computational hardness, the existence of such algorithms for certain {\sf NP}-complete problems has also been ruled out under reasonable assumptions: \textsc{$k$-SetCover} \cite{cck+17,CL19,KLM19,Lin19,KL21,LRSW23a}, $k$-{\sc SetIntersection} \cite{lin2018parameterized,BKN21}, $k$-{\sc Steiner Orientation} \cite{wlodarczyk2020parameterized}, {\sc Max-$k$-Coverage}~\cite{manurangsi2020tight}, \textsc{$k$-SVP, $k$-MinDistanceProblem} and related problems~\cite{manurangsi2020tight,BBE+21,BCGR23}. 
    
    An exciting recent line of work \cite{cck+17, lin21, LRSW22, karthik2022almost, CFL+23, lin2023improved} shows that approximating $k$-{\sc Clique} is not FPT under Gap-ETH, ETH, and ${\sf W[1]}\ne {\sf FPT}$. 
\end{itemize}
We refer to the survey by Feldmann, Karthik, Lee, and Manurangsi~\cite{feldmann2020survey} for a detailed discussion.

\paragraph{The Quest for Parameterized PCP-Type Theorems.}
Despite all the recent progress in the study of parameterized inapproximability, the reductions presented in these papers are often ad-hoc and tailored to the specific problems in question. Obtaining a unified and powerful machinery for parameterized inapproximability, therefore, becomes increasingly important.

A good candidate is to \emph{establish a parameterized PCP-type theorem}.
The PCP theorem~\cite{arora1998probabilistic,arora1998proof,dinur2007pcp}, a cornerstone of modern complexity theory, gives a polynomial time reduction from an {\sf NP}-hard problem like {\sc 3SAT} to a gap version of {\sc 3SAT} where the goal is to distinguish satisfiable instances from those for which every assignment fails to satisfy a $\gamma$ fraction of clauses for some absolute constant $\gamma > 0$. 
This then serves as the starting point for a large body of inapproximability results for fundamental problems, including constraint satisfaction, graph theory, and optimization. 

As discussed in~\cite{feldmann2020survey}, the current situation in the parameterized world is similar to that of the landscape of the traditional hardness of approximation \emph{before} the celebrated PCP theorems was established.
Given the similarity, the following folklore open problem has been recurring in the field of parameterized inapproximability:
\begin{center}\it
Can we establish a PCP-type theorem in the parameterized complexity theory?
\end{center}

In light of its rising importance, Lokshtanov, Ramanujan, Saurabh, and Zehavi~\cite{LRSZ20} formalized and entitled the above question as \emph{Parameterized Inapproximability Hypothesis (PIH)}.
Here we present the following reformulation\footnote{The original statement of PIH in \cite{LRSZ20} replaces the runtime bound by $\mathsf{W}[1]$-hardness, and the reformulation by \cite{feldmann2020survey} suffices for applications.} of PIH due to \cite{feldmann2020survey}:

\begin{hypothesis}[Parameterized Inapproximability Hypothesis]\label{hypo:pih_intro}
There is an absolute constant\footnote{The exact constant here is not important. Starting from a constant $\eps>0$, one can boost it to $1-\eta$ for any constant $\eta>0$ by standard reductions.} $\varepsilon>0$, such that no fixed parameter tractable algorithm which, takes as input a 2CSP $G$ with $k$ variables of size-$n$ alphabets, can decide whether $G$ is satisfiable or at least $\varepsilon$ fraction of constraints must be violated.
\end{hypothesis}

Similar to the PCP theorem, PIH, if true, serves as a shared beginning for results in parameterized hardness of approximation: $k$-{\sc Clique}, $k$-{\sc SetCover}, {\sc $k$-ExactCover}~\cite{guruswami2023baby}, {\sc Shortest Vector} \cite{BBE+21,BCGR23}, {\sc Direct Odd Cycle Transversal}~\cite{LRSZ20}, {\sc Determinant Maximization} and {\sc Grid Tiling}~\cite{ohsaka2022parameterized}, Baby PIH~\cite{guruswami2023baby}, $k$-{\sc MaxCover}~\cite{KLM19}, and more. 

Prior to our work, PIH was only proved under the Gap-ETH assumption, the gap version of ETH. Since there is an inherent gap in Gap-ETH, the result can be obtained by a simple gap-preserving reduction~(see, e.g., \cite{feldmann2020survey}). Indeed, it is often recognized that gap-preserving reductions are much easier than gap-producing reductions~\cite{feige1996interactive}. A more desirable result is, analogous to the PCP theorem, to create a gap from a gap-free assumption:
\begin{center}
    \textit{Can we prove PIH under an assumption without an inherent gap?}
\end{center}

\subsection{Our Results}\label{sec:our_results}

We answer the above question in the affirmative by proving the \emph{first} result to base PIH on a gap-free assumption. We consider the famous Exponential Time Hypothesis (ETH) \cite{IP01}, a fundamental gap-free hypothesis in the modern complexity theory and a weakening of the Gap-ETH assumption.
\begin{hypothesis*}[Exponential Time Hypothesis (ETH), Informal]
Solving {\sc 3SAT} needs $2^{\Omega(n)}$ time.
\end{hypothesis*}
Our main theorem can be stated concisely as:
\begin{theorem}[Main]\label{thm:main}
ETH implies PIH.
\end{theorem}

In \Cref{thm:main_formal}, we provide a quantitative version of \Cref{thm:main} that presents an explicit runtime lower bound of $f(k)\cdot n^{\Omega\pbra{\sqrt{\log \log k}}}$ for the problem in \Cref{hypo:pih_intro} under ETH.

As a byproduct of the above quantitative bound, we have the following probabilistic checkable proof version of the main theorem (see \Cref{thm:main_formal_pcp} for the full version). This can be seen as a PCP theorem in the parameterized world where the proof length depends only on $k$ (which is supposed to be a small growing parameter), but the alphabet size is the significantly growing parameter. The runtime of the PCP verifier is in {\sf FPT}.

\begin{theorem}\label{thm:main_pcp_version}
For any integer $k\ge1$, {\sc 3SAT} has a PCP verifier which can be constructed in time $f(k)\cdot|\Sigma|^{O(1)}$ for some computable function $f$,  makes two queries on a proof with length $2^{2^{O(k^2)}}$ and alphabet size $|\Sigma|=2^{O(n/k)}$, and has completeness 1 and soundness $1-\frac{1}{9600}$.
\end{theorem}

As mentioned, PIH serves as a unified starting point for many parameterized inapproximability results. Below, we highlight some new ETH-hardness of approximation for fundamental parameterized problems obtained by combining our result and existing reductions from PIH.

\paragraph*{Application Highlight: $k$-ExactCover.} 
$k$-{\sc ExactCover} (also known as $k$-{\sc Unique Set Cover}) is a variant of the famous $k$-{\sc SetCover} problem. In the $\rho$-approximation version of this problem, denoted by $(k, \rho\cdot k)$-{\sc ExactCover}, we are given a universe $U$ and a collection $\mathcal S$ of subsets of $U$. The goal is to distinguish the following two cases.
\begin{itemize}
    \item There exists $k$ \emph{disjoint} subsets that cover the whole universe $U$, i.e., the union of the $k$ subsets is exactly the whole universe $U$.
    \item Any $\rho\cdot k$ subsets of $\Scal$ cannot cover $U$. 
\end{itemize}
Here, the parameter is the optimum $k$. 
Note that $k$-{\sc ExactCover} is an easier problem than $k$-{\sc SetCover} due to the additional disjointness property, proving computational hardness even harder.
On the positive side, this additional structure also makes {\sc $(k, \rho\cdot k)$-ExactCover} an excellent proxy for subsequent reductions.
We refer interested readers to previous works for details~\cite{ABSS97, manurangsi2020tight}. 

For constant $\rho>1$, the hardness of $(k, \rho\cdot k)$-{\sc ExactCover} was only proved under assumptions with inherent gaps~\cite{manurangsi2020tight, guruswami2023baby}, imitating the reduction in the non-parameterized world \cite{feige1998threshold}. It was still a mystery whether we could derive the same result under a weaker and gap-free assumption. 
Combining its PIH hardness (see e.g., ~\cite{guruswami2023baby}\footnote{\cite{guruswami2023baby} proves the hardness of $(k, \rho\cdot k)$-{\sc ExactCover} under a weaker version of PIH, namely, Average Baby PIH with rectangular constraints.}) with our main theorem~(\Cref{thm:main}), we prove the first inapproximability for $k$-{\sc ExactCover} under a gap-free assumption.
\begin{corollary}\label{cor:exactcover}
    Assuming ETH, for any absolute constant $\rho\ge 1$, no FPT algorithm can decide $(k, \rho\cdot k)$-{\sc ExactCover}.
\end{corollary}

\paragraph*{Application Highlight: Directed Odd Cycle Transversal.} 
Given a directed graph $D$, its directed odd cycle transversal, denoted by ${\rm DOCT}(D)$, is the minimum set $S$ of vertices such that deleting $S$ from $D$ results in a graph with no directed odd cycles. The $\rho$-approximating version of the directed odd cycle transversal problem, denoted by $(k, \rho\cdot k)$-DOCT, is to distinguish directed graphs $D$ with ${\rm DOCT}(D)\le k$, from those with ${\rm DOCT}(D)\ge \rho\cdot k$. 
The parameter of this problem is the optimum $k$. 
This problem is a generalization of several well studied problems including {\sc Directed Feedback Vertex Set} and {\sc Odd Cycle Transversal}.
For a brief history of this problem, we refer to the previous work~\cite{LRSZ20}. 
In their work, the authors prove the following hardness of $(k, \rho\cdot k)$-DOCT.
\begin{theorem}[\cite{LRSZ20}]
    Assuming PIH, for some $\rho\in (1, 2)$, no FPT algorithm can decide $(k, \rho\cdot k)$-DOCT.
\end{theorem}
Combining the theorem above with \Cref{thm:main}, we establish the first hardness of $(k, \rho\cdot k)$-DOCT under a gap-free assumption.
\begin{corollary}\label{cor:doct}
    Assuming ETH, for some $\rho\in (1, 2)$, no FPT algorithm can decide $(k, \rho\cdot k)$-DOCT.
\end{corollary}

\subsection{Overview of Techniques}\label{sec:overview}

To prove our main theorem~(\Cref{thm:main}), we present an efficient reduction from {\sc 3SAT} formulas to parameterized CSPs of $k$ variables with a constant gap. 

To construct such a reduction, we follow the widely-used paradigm for proving PCP theorems~\cite{ arora1998probabilistic,arora1998proof, vjs20}. 
Via this approach, we first arithmetize {\sc 3SAT} into an intermediate CSP (usually a constant-degree polynomial system in the literature) with $k$ variables and alphabet $\Sigma_1$. 
Then, we decide on a locally testable and correctable code $C\colon\Sigma_1^{k}\to \Sigma_2^{k'}$~(e.g., the quadratic code~\cite{arora1998proof}, the Reed-Muller code~\cite{arora1998probabilistic}, or the long code~\cite{vjs20}), and treat the proof $\pi$ as an encoding of some assignment $\sigma\colon[k]\to \Sigma_1$ to the intermediate CSP (viewed as a vector in $\Sigma_1^k$). 
Leveraging the power of the local testability and correctability of $C$, we will check whether the input proof is (close to) the encoding of an assignment that satisfies the intermediate CSP. 

\paragraph{Our Plan.} 
To follow the outline above and also factor in the runtime and the parameter blowup, our plan is as follows:
\begin{enumerate}
    \item\label{itm:intro_overview_1} First, we need to design an appropriate intermediate parameterized CSP problem, which has some runtime lower bound under ETH. In the parameterized setting, the number of variables is a small parameter $k$, while the alphabet $|\Sigma_1|$ holds the greatest order of magnitude.
    \item\label{itm:intro_overview_2} Second, we need to construct an error correcting code $C$, which can be used to encode a solution of the intermediate CSP, and allows us to locally check its satisfiability.
    Here the efficiency of the code is also measured in a parameterized sense that the alphabet size is polynomial in $|\Sigma_1|$, and codeword length can be arbitrary in $k$ but independent of $|\Sigma_1|$.
\end{enumerate}

However, the plan above confronts the following basic obstacle. 
The constructions in proving the PCP theorems usually require the proof length $ |\pi| = |\Sigma_1|^{\Omega(k)}$. 
On the other hand, as illustrated in \Cref{itm:intro_overview_2} above, we must eliminate $|\Sigma_1|$ in the proof length to make sure that the reduction is FPT.

\paragraph{Vectorization.} 
We bypass this obstacle by applying \emph{vectorization}, an idea also used in \cite{lin2023improved}. 
In detail, we enforce the alphabet $\Sigma_1$ to be a vector space $\Fbb^d$, where $\Fbb$ is a field of constant size. 
In this way, an assignment $\sigma\in\Sigma_1^k=(\Fbb^d)^k$ can be viewed as $d$ parallel sub-assignment in $\Fbb^k$.
Thus, if we have a good code $C\colon\Fbb^k\to \Fbb^{k'}$ that tests the validity of a sub-assignment, we can encode $\sigma$ by separately encoding each sub-assignment and combining them as an element in $(\Fbb^{k'})^d=(\Fbb^d)^{k'}$.
Since $|\Fbb|$ is a constant, this makes $k'$ dependent only on $k$ but not on the whole alphabet $\Sigma_1 = \Fbb^d$. 

Guided by the vectorization idea, we aim to design an ETH-hard intermediate CSP problem where the alphabet is a vector space.
Furthermore, to facilitate the construction of the error correcting code $C$ in the second step, we also hope that there are appropriate restrictions on the constraints of this intermediate CSP problem. 
The constraints should be neither too restrictive (which loses the ETH-hardness) nor too complicated (which hinders an efficient testing procedure). 
Following these intuitions, we define the following \emph{Vector-Valued CSPs} as our intermediate problem.

\paragraph{Vector-Valued CSPs.} 
Vector-Valued CSPs (\Cref{def:vcsp}) are CSPs with the following additional features:
\begin{itemize}
    \item The parameter $k$ is the number of variables. 
    \item The alphabet is a vector space $\Fbb^d$, where $\Fbb$ is a finite field of characteristic two and constant size and $d$ holds the greatest order of magnitude.
    \item Each constraint is either a (coordinate-wise) parallel constraint, or a linear constraint, where:
    \begin{itemize}
        \item A parallel constraint is defined by a sub-constraint $\Pi^{sub}\colon\Fbb\times \Fbb\to \{0,1\}$ and a subset of coordinates $Q\subseteq [d]$. It checks $\Pi^{sub}$ for every coordinate in $Q$ of the vector assignments.
        \item A linear constraint is defined by a matrix $M$. It enforces that two vector assignments satisfy a linear equation specified by $M$.
    \end{itemize}
    \item Each variable is related to at most one parallel constraint.
\end{itemize}
We emphasize that vector-valued CSPs will become fixed parameter tractable if all constraints are linear (resp., parallel).
In detail, one can handle linear constraints by efficiently solving a system of linear equations, or handle parallel constraints by brute force enumeration individually for each coordinate.
However, due to our reduction, one cannot  solve vector-valued CSPs with both constraint types efficiently  under ETH. 

\paragraph{{\sc 3SAT} to Vector-Valued CSPs.}
In \Cref{thm:3satCSP}, we establish the ETH-hardness of vector-valued CSP instances by a series of standard transformations.

First, we partition the clauses and variables of a {\sc 3SAT} formula respectively into $k$ parts. Each of the $2k$ parts is then built as a CSP variable, which takes assignments of that part of clauses/variables. The alphabet is therefore a vector space.

Then, we impose constraints between clause parts and variable parts. Each constraint is a conjunction of clause validity and clause-variable consistency. These constraints ensure that the $2k$ partial assignments correspond to a global satisfying assignment to the original {\sc 3SAT} formula.

However, the constraints above are neither parallel nor linear. To make them parallel, we first appropriately split constraints, then duplicate each variable into several copies and spread out its constraints. 
After this procedure, each variable is related to exactly one constraint, and each constraint is the same sub-constraint applied in a matching way on the $d$ coordinates of the related vector-variables. We can thus permute the $d$ coordinates of each variable accordingly and obtain the parallel constraint form we desire. In addition, we also need to check the (permuted) consistency between different duplicates. 
These checks can be done using permuted equality constraints, which are special forms of linear constraints. 

\paragraph{Vector-Valued CSPs to Constant-Gap CSPs.}
In \Cref{prop:reduction}, we construct another FPT reduction from a vector-valued CSP to a general constant-gap parameterized CSP in three steps.
\begin{itemize}
    \item First, we split the vector-valued CSP instance into two by partitioning the constraints into a linear part and a parallel part. 
    \item Next, for each of the two sub-instances, we construct a randomized verifier to check whether all constraints in it are satisfied. 
    The verifier takes as input a parallel encoding of a solution.
    It then flips random coins, makes a constant number of queries based on the randomness, and decides whether to accept the input proof or not based on the query result. The verifier will have a constant soundness and a constant proximity parameter.
    In the traditional complexity theory, such verifiers are also known as Probabilistic Checkable Proof of Proximity (PCPP) verifiers. 
    
    In our proof, the verifier is designed separately for linear constraints and parallel constraints. The consistency of the two verifiers is guaranteed via a unified parallel Walsh-Hadamard encoding of the solution, shared by both verifiers.
    \item Finally, we obtain a constant-gap CSP instance by a standard reduction from probabilistic checkable proof verifiers to CSPs.
\end{itemize}

The proof is then completed by combining the two reductions above.
The crux of our proof is the design of the PCPP verifiers. Below, we present high-level descriptions of this part.

\paragraph{PCPPs for Vector-Valued CSPs with Parallel Constraints.} 
Fix a vector-valued CSP instant $G$ with parallel constraints only.
The key observation in designing PCPPs for $G$ is that, though the parallel sub-constraints can be arbitrary, different coordinates of the vector-variables are independent. 
Let $k$ be the number of variables in $G$ and let $d$ be the dimension of the vector-variables.

Following the observation above, we can split $G$ into $d$ sub-instances $G_1,\dots, G_d$ with respect to the $d$ coordinates. Each $G_i$ is a CSP instance with $k$ variables and alphabet $\mathbb F$. A vector-valued assignment $\sigma$ satisfies $G$ iff the sub-assignment of $\sigma$ on the $i$-th coordinate satisfies $G_i$ for each $i \in [d]$. 

After splitting, the alphabet of each $G_i$ is only $\Fbb$. We can thus follow the classical construction~\cite{arora1998proof, arora2009computational} of PCPPs to construct a verifier $A_i$ to efficiently and locally check the satisfiability of $G_i$.
In addition, since every vector-variable is related to at most one parallel constraint in $G$, the number of distinct sub-instances among $G_1,\ldots,G_d$ depends only on $k$, not on $d$.
This allows us to combine $A_1,\dots,A_d$ into a single verifier $A$ that works over the original alphabet $\Fbb^d$ with blowup dependent only on $k$, not on $d$.
See \Cref{sec:parallel-pcpp} for details.

\paragraph{PCPPs for Vector-Valued CSPs with Linear Constraints.} To design a verifier for linear constraints, we leverage the power of the Walsh-Hadamard code to decode any linear combinations of the messages. 
Fix a vector-valued CSP instance $G$ with linear constraints only. For each linear constraint $e=(u_e,v_e)\in E$, we further denote its form by $\indicator_{u_e = M_ev_e}$.

To test the conjunction of all linear constraints, it is natural to consider the linear combination of these constraints. In detail, we pick independently random $\lambda_1,\dots,\lambda_{|E|}\in \Fbb$, and test whether:
$$
\sum_{e\in E} \lambda_e u_e = \sum_{e\in E} \lambda_e\cdot M_ev_e
$$
By the random subsum principle~\cite{arora2009computational}, if any one of the linear constraints is violated, the equation above does not hold with high probability.

Following this idea, we introduce auxiliary variables $z_{v,e}$ for each variable $v$ and constraint $e$, which is supposed to be $M_e v$. We set up the parallel version of the Walsh-Hadamard code over the assignments to the variables in $G$ and the auxiliary variables $z_{v,e}$. 
In this way, we can decode both the LHS and RHS of the equation above by two queries on the Walsh-Hadamard code, and then check whether the equation holds.

We need extra testing procedures to ensure $z_{v,e}$ equals $M_e v$. This is again achieved by the random subsum principle. See \Cref{sec:linear-pcpp} for details.

\subsection{Discussions}

Here we discuss related works and propose future directions.

\paragraph{Related Works.}
As mentioned above, prior to our work, PIH was only known to hold under Gap-ETH \cite{cck+17,DM18}.
The techniques there do not apply here since their proofs rely on an inherent gap from the assumption, which ETH does not have. 

Using a different approach, Lin, Ren, Sun, and Wang \cite{LRSW22, lin2023improved} proposed to prove PIH via a strong lower bound for constant-gap {\sc $k$-Clique}. This is reminiscent of \cite{BellareGS98}, where the {\sf NP}-hardness of constant-gap {\sc Clique} leads to a free-bit PCP. However, the construction in \cite{BellareGS98} does not apply in the parameterized setting since the proof length will be too long.
In addition, the framework of \cite{lin2023improved} only designs a weaker variant of PCPP, which can only locally test the validity of a single constraint rather than the conjunction of all constraints. Moreover, the boosting from weak PCPPs to standard PCPPs seems to meet an information-theoretic barrier from locally decodable codes.
In contrast, we successfully design PCPPs for special CSPs in this work, which is based on a key observation that CSPs remain ETH-hard even when the variables are vector-valued and the constraints are either parallel or linear.

Furthermore, a recent work by Guruswami, Ren, and Sandeep \cite{guruswami2023baby} established a weaker version of PIH called Baby PIH, under ${\sf W[1]\ne FPT}$. However, they also gave a counterexample to show that the basic direct product approach underlying their reduction is not strong enough to establish PIH.

\paragraph{Future Directions.} 
We highlight some interesting open directions.
\begin{itemize}
    \item Starting from PIH and by our work, many previous parameterized hardness of approximation results can now be based on ETH (see \Cref{cor:exactcover} and \Cref{cor:doct} as representatives). However, there are still many basic problems whose parameterized inapproximability remains unknown, e.g., {\sc Max} $k$-{\sc Coverage} and $k$-{\sc Balanced Biclique}~\cite{cck+17,feldmann2020survey}.
    
    Can we discover more ETH-based parameterized inapproximability results? Since there is already a gap in PIH, we expect that reducing PIH to other parameterized problems would be easier than reducing directly from gap-free {\sc 3SAT}. 
    \item We have presented a gap-producing reduction from ETH to PIH. It is natural to ask whether we can prove PIH under the minimal hypothesis ${\sf W[1]\ne FPT}$.
    Our paper constructs an FPT reduction from vector-valued CSPs to gap CSPs. 

    We remark that our vector-valued CSP instances are closely related to an $\sf M[1]$-complete problem {\sc Mini-3SAT} \cite{fellows2003blow} where $\sf M[1]$ is an intermediate complexity class between $\sf FPT$ and $\sf W[1]$.
    Thus, unless $\sf M[1]=W[1]$, our proof may not be directly generalized to prove PIH under $\sf W[1]\ne FPT$.
    We refer interested readers to \cite{chen2007isomorphism} for a detailed discussion of these complexity classes and hierarchies.
    
\end{itemize}

\paragraph{Paper Organization.}
In \Cref{sec:prelim}, we define necessary notation and introduce useful tools from the literature. Then, the paper is organized in a modular manner. First, in \Cref{sec:structure}, we present the proof of our main result with the proofs of technical lemmas deferred to later sections. Then, in \Cref{sec:3sat-to-csp}, we show how to obtain a vector-valued CSP instance with desired structures from {\sc 3SAT} as needed in \Cref{sec:structure}. Next, in \Cref{sec:parallel-pcpp}, we design the probabilistic verifier for parallel constraints in the CSP instance, another building block needed in \Cref{sec:structure}.
Finally, in \Cref{sec:linear-pcpp}, we give the probabilistic verifier for linear constraints in the CSP instance, the last missing piece of \Cref{sec:structure}.

\section{Preliminaries}\label{sec:prelim}

For a positive integer $n$, we use $[n]$ to denote the set $\cbra{1,2,\ldots,n}$.
We use $\log$ to denote the logarithm with base $2$.
For an event $\Ecal$, we use $\indicator_\Ecal$ as the indicator function, which equals $1$ if $\Ecal$ happens and equals $0$ if otherwise.
For disjoint sets $S$ and $T$, we use $S\dot\cup T$ to denote their union while emphasizing $S\cap T=\emptyset$.
For a prime power $q=p^t$ where $p$ is a prime and $t\ge1$ is an integer, we use $\Fbb_q$ to denote the finite field of order $p^t$ and characteristic $p$.

We use superscript $\top$ to denote vector / matrix transpose. For two vectors $u,v\in\Fbb^d$, we use $\abra{u,v}$ to denote their inner product which equals $u^\top v$ (or $v^\top u$).
For two matrices $A,B\in\Fbb^{d\times d}$, we use $\abra{A,B}=\sum_{i,j\in[d]}A_{i,j}B_{j,i}$ to denote their inner product.

Throughout the paper, we use $O(\cdot),\Theta(\cdot),\Omega(\cdot)$ to hide absolute constants that do not depend on any other parameter.

\subsection{(Parameterized) Constraint Satisfaction Problems}\label{sec:kcsp}

\paragraph*{CSP.}
In this paper, we only focus on constraint satisfaction problems (CSPs) of arity two. Formally, a CSP instance $G$ is a quadruple $(V, E, \Sigma, \{\Pi_{e}\}_{e\in E})$, where:
\begin{itemize}
    \item $V$ is for the set of variables.
    \item $E$ is for the set of constraints.
    Each constraint $e=\cbra{u_e,v_e} \in E$ has arity $2$ and is related to two distinct variables $u_e,v_e\in V$.

    The \emph{constraint graph} is the undirected graph on vertices $V$ and edges $E$. Note that we allow multiple constraints between a same pair of variables and thus the constraint graph may have parallel edges.
    \item $\Sigma$ is for the alphabet of each variable in $V$. For convenience, we sometimes have different alphabets for different variables and we will view them as a subset of a grand alphabet $\Sigma$ with some natural embedding.
    \item $\{\Pi_{e}\}_{e\in E}$ is the set of constraint validity functions. 
    Given a constraint $e\in E$, the validity function $\Pi_e(\cdot,\cdot)\colon\Sigma\times\Sigma\to\bin$ checks whether the constraint $e$ between $u_e$ and $v_e$ is satisfied.
\end{itemize}
We use $|G|=(|V|+|E|)\cdot|\Sigma|$ to denote the \emph{size} of a CSP instance $G$.

\paragraph*{Assignment and Satisfiability Value.}
An \emph{assignment} is a function $\sigma\colon V\to \Sigma$ that assigns each variable a value in the alphabet. 
The \emph{satisfiability value} for an assignment $\sigma$, denoted by $\val(G,\sigma)$, is the fraction of constraints satisfied by $\sigma$, i.e.,
$\val(G, \sigma)=\frac{1}{|E|}\sum_{e\in E} \Pi_e(\sigma(u_e),\sigma(v_e))$.
The satisfiability value for $G$, denoted by $\val(G)$, is the maximum satisfiability value among all assignments, i.e., $\val(G) = \max_{\sigma\colon V\to \Sigma} \val(G, \sigma)$. 
We say that an assignment $\sigma$ is a \emph{solution} to a CSP instance $G$ if $\val(G,\sigma) = 1$, and $G$ is \emph{satisfiable} iff $G$ has a solution.

When the context is clear, we omt $\sigma$ in the description of a constraint, i.e., $\Pi_e(u_e,v_e)$ stands for $\Pi(\sigma(u_e),\sigma(v_e))$.

\paragraph*{Parameterization and Fixed Parameter Tractability.} 
For an instance $G$, the \emph{parameterization} refers to attaching the parameter $k := |V|$ (the size of the variable set) to $G$ and treating the input as a $(G, k)$ pair.
We think of $k$ as a growing parameter that is much smaller than the instance size $n := |G|$. 
A promise problem $L_{\sf yes}\cup L_{\sf no}$ is \emph{fixed parameter tractable (FPT)} if it has an algorithm which, for every instance $G$, decides whether $G\in L_{\sf yes}$ or $G\in L_{\sf no}$ in $f(k)\cdot n^{O(1)}$ time for some computable function $f$. 

\paragraph*{FPT Reduction.} 
An \emph{FPT reduction} from $L_{\sf yes}\cup L_{\sf no}$ to $L'_{\sf yes}\cup L'_{\sf no}$ is an algorithm $\mathcal A$ which, on every input $G=(V, E, \Sigma, \{\Pi_{e}\}_{e\in E})$ outputs another instance $G'=(V', E', \Sigma', \{\Pi_{e}'\}_{e\in E'})$ such that:
\begin{itemize}
    \item \textsc{Completeness.} If $G\in L_{\sf yes}$, then $G'\in L'_{\sf yes}$.
    \item \textsc{Soundness.} If $G\in L_{\sf no}$, then $G'\in L'_{\sf no}$.
    \item \textsc{FPT.} There exist universal computable functions $f$ and $g$ such that  $|V'|\le g(|V|)$ and the runtime of $\mathcal A$ is bounded by $f(|V|)\cdot |G|^{O(1)}$.
\end{itemize}

\paragraph*{$\eps$-Gap $k$-CSP.} 
We mainly focus on the gap version of the parameterized CSP problem. Formally, an \gapkcsp{} problem needs to decide whether a given CSP instance $(G,|V|)$ with $|V|=k$ satisfies $\val(G)=1$ or $\val(G)<1-\varepsilon$. 
The exact version is equivalent to {\sc $0$-gap $k$-CSP}.

\paragraph*{Parameterized Inapproximability Hypothesis (PIH).} 
\emph{Parameterized Inapproximability Hypothesis (PIH)}, first\footnote{As noted in \cite{LRSZ20}, prior to their work, this hypothesis was already informally stated by quite a few researchers as a natural formulation of the PCP theorem in parameterized complexity.} formulated by Lokshtanov, Ramanujan, Saurabh, and Zehavi \cite{LRSZ20}, is a central conjecture in the parameterized complexity theory, which, if true, serves as a parameterized counterpart of the celebrated PCP theorem. Below, we present a slight reformulation of PIH, asserting fixed parameter intractability (rather than $W[1]$-hardness specifically) of gap CSP.

\begin{hypothesis}[PIH]\label{hypo:pih}
For an absolute constant $0<\varepsilon<1$, no FPT algorithm can decide $\gapkcsp$.
\end{hypothesis}

\paragraph*{Exponential Time Hypothesis (ETH).}
\emph{Exponential Time Hypothesis (ETH)}, first proposed by Impagliazzo and Paturi \cite{IP01}, is a famous strengthening of the $\mathsf{P}\neq\mathsf{NP}$ hypothesis and provides a foundation for fine-grained understandings in the modern complexity theory.

\begin{definition}[{\sc 3SAT}]\label{def:3sat}
A 3CNF formula $\varphi$ on $n$ Boolean variables is a conjunction of $m$ clauses, where each clause is a disjunction of three literals and each literal is a variable or its negation.
The goal of the {\sc 3SAT} problem is to decide whether $\varphi$ is satisfiable or not.
\end{definition}

The original ETH is stated in the general {\sc 3SAT} problem.
In this paper, for convenience, we use the following variant due to the sparsification lemma~\cite{IPZ01} and Tovy's reduction \cite{Tov84}, which gives {\sc 3SAT} additional structure.

\begin{hypothesis}[ETH]\label{hypo:eth}
No algorithm can decide {\sc 3SAT} within runtime $2^{o(n)}$, where additionally each variable is contained in at most four clauses and each clause contains exactly three distinct variables.\footnote{We say a variable $x$ is contained in a clause $C$ if the literal $x$ or $\neg x$ appears in $C$.}
\end{hypothesis}

\subsection{Parallel Walsh-Hadamard Code}\label{sec:hadamard_code}

As mentioned in \Cref{sec:overview}, the key step to bypass the obstacle in previous constructions is vectorization and parallel encoding of an error correcting code. In this paper, we only consider the parallelization of the famous \emph{Walsh-Hadamard code}, a classic error correcting code that is locally testable and correctable. First, we recall standard notions in coding theory.

Given two words (aka strings) $x,y\in\Sigma^K$ and same length $K$, their \emph{relative distance} $\Delta(x,y)$ is the fraction of coordinates that they differ, i.e., $\Delta(x,y) = \frac1K|\{i\in [K]\colon x_i\ne y_i\}|$. 
We say $x\in \Sigma^K$ is $\delta$-far (resp., $\delta$-close) from a set of words $S\subseteq \Sigma^K$ if $\Delta(x,S) := \min_{y\in S} \Delta(x,y)\ge \delta$ (resp., $\le \delta$).

\begin{definition}[Error Correcting Codes (ECCs)]\label{def:ecc}
An error correcting code is the image of the encoding map $C\colon\Sigma_1^k\to \Sigma_2^K$ with message length $k$, codeword length $K$. We say that the ECC has a relative distance $\delta$ if $\Delta(C(x),C(y))\ge \delta$ holds for any distinct $x,y\in\Sigma_1^k$. 
We use $\rm{Im}(C)$ to denote the codewords of $C$.
\end{definition}

\begin{definition}[Parallel Walsh-Hadamard Code]\label{def:whcode}
Let $\Fbb$ be a finite field and $(a_1,a_2,\dots,a_k)\in (\Fbb^d)^k$ be a tuple of $k$ vectors in $\Fbb^d$. 
We view it as a matrix $A = (a_1,a_2,\dots,a_k)\in \mathbb F^{d\times k}$ where the $i$-th column is the vector $a_i$.

The parallel Walsh-Hadamard encoding $\pwh(A)$ of $A$ is a codeword indexed by $\Fbb^k$ where each entry is a vector in $\Fbb^d$.
Alternatively, $\pwh(A)$ is a function mapping $\Fbb^k$ to $\Fbb^d$ that enumerates linear combinations of the column vectors of $A$. 
Formally, for each $b\in \mathbb F^k$, we have $\pwh(A)[b]= Ab$.
\end{definition}

We remark that the parallel Walsh-Hadamard code is also known as interleaved Hadamard code and linear transformation code \cite{gopalan2009list,dinur2008decodability}.

In the notation of \Cref{def:ecc}, $\pwh$ has $\Sigma_1=\Fbb^d$, $\Sigma_2=\Fbb^d$, and $K=|\Fbb|^k$.
Note that when $d=1$, the parallel Walsh-Hadamard code coincides with the standard Walsh-Hadamard code. It is clear that $\pwh$ has the relative distance $\delta= 1-\frac1{|\Fbb|}$, which is at least $\frac12$ since $|\Fbb|\ge2$ holds always.

\paragraph*{Local Testability and Correctability.}
Fix a word $w\in (\mathbb F^d)^{\mathbb F^k}$ and treat it as a map from $\mathbb F^k$ to $\mathbb F^d$.
To test whether $w$ is close to a codeword of $\pwh$, we perform the famous \emph{BLR test} \cite{blum1993self}, which samples uniformly random $a,b\in \mathbb F^k$ and accept if $w[a]+w[b]=w[a+b]$ by three queries to $w$. 
The following theorem establishes the soundness of this test.

\begin{theorem}\label{thm:pwh_test}
If $\Pr_{a,b\in \mathbb F^k}\sbra{w[a]+w[b]=w[a+b]}\ge 1-\varepsilon$, then $\Delta(x,{\rm Im}(\pwh))\le6\varepsilon$.
\end{theorem}
\begin{proof}
If $\eps<1/6$, we apply the bound from \cite[Theorem 3]{goldreich2016lecture} and obtain $\Delta(x,{\rm Im}(\pwh))\le2\varepsilon\le6\eps$.
Otherwise $\eps\ge1/6$ and we naturally have $\Delta(x,{\rm Im}(\pwh))\le1\le6\varepsilon$.
\end{proof}

Assume $w$ is $\eta$-close to an actual codeword $w^*$ of $\pwh$.
To obtain the value of $w^*[x]$ for some $x\in\Fbb^k$, we can draw a uniform $a\in \mathbb F^k$ and compute $w[x+a]-w[a]$ by two queries. 
The following fact concerns the soundness of this procedure.

\begin{fact}\label{fct:pwh_correct}
If $w$ is $\eta$-close to some $w^*\in {\rm Im}(\pwh)$, then $\Pr_{a\in\Fbb^k}\sbra{w[x+a]-w[a]=w^*[x]}\ge1-2\eta$.
\end{fact}

\subsection{Probabilistic Checkable Proofs with Proximity}\label{sec:pcpp} 

Probabilistic Checkable Proofs of Proximity (PCPP, also known as assignment testers) \cite{BGH05,dinur2006assignment} are essential gadgets when proving the PCP theorem~\cite{arora1998probabilistic,dinur2007pcp, arora2009computational}.
There, the gadget is used to verify whether a set of Boolean variables is close to a solution of a formula given by a circuit.

In this paper, we reformulate PCPP under the parameterized regime. Our reformulation is compatible with the parallel encoding. To conveniently combine different PCPPs, we specialize PCPPs into their $\pwh$-based constructions.\footnote{The choice of encoding is typically abstracted out in standard definitions of PCPPs.} Formally, we define the following \emph{parallel probabilistic checkable proofs with proximity (PPCPP)}.

\begin{definition}[$(q,\delta,\varepsilon,f,g)$-PPCPPs]\label{def:pcpp}
    Let $f$ and $g$ be two computable functions.
    Given a finite field $\Fbb$ and a CSP instance $G = (V,E,\Sigma,\{\Pi_e\}_{e\in E})$ where $\Sigma=\Fbb^d$. Its $(q,\delta,\eps,f,g)$-PPCPP is a randomized verifier $A$ with the following workflow: Recall that $k=|V|$ is the parameter of the CSP instance $G$.
    \begin{itemize}
        \item $A$ takes as input two blocks of proofs $\pi_1\circ \pi_2$ with alphabet $\Fbb^d$, where:
        \begin{itemize}
            \item $\pi_1$ has length $|\Fbb|^k$ with entries indexed by vectors in $\Fbb^k$, which is supposed to be the parallel Walsh-Hadamard encoding of some assignment to $V$.
            \item $\pi_2$ has length at most $f(k)$. It is an auxiliary proof enabling an efficient verification procedure. 
        \end{itemize}
        \item $A$ chooses a uniform $r\in [R_A]$, where $R_A$ is at most $g(k)$, queries at most $q$ positions in $\pi_1\circ \pi_2$ based on $r$, and decides to accept or reject the proof after getting the query result.
        \item The list of queries made by $A$ can be generated in time at most $h(k)\cdot |G|^{O(1)}$ for some computable function $h$.
    \end{itemize}
    The verifier $A$ has the following properties.
    \begin{itemize}
        \item \textsc{Completeness.} For every solution $\sigma$ of $G$, there exists a $\pi_2$ such that $\Pr[A\ {\rm accepts}\ \pwh(\sigma)\circ \pi_2]=1$, where we treat an assignment $\sigma\colon V\to \Fbb^d$ as a vector in $(\Fbb^{d})^{|V|}$.
        \item \textsc{Soundness.} If $\Pr[A\ {\rm accepts}\ \pi_1\circ \pi_2]\ge 1-\varepsilon$, there exists some solution $\sigma$ of $G$ such that $\Delta(\pi_1, \pwh(\sigma))\le \delta$.
    \end{itemize}
\end{definition}

Intuitively, PPCPPs check whether $\pi_1$ is close to the Walsh-Hadamard encoding of some solution of $G$.
Like the traditional PCPP, parallel PCPPs are also tightly connected with CSPs. The following standard reduction establishes the connection.

\begin{definition}[Reduction from PPCPPs to CSPs]\label{def:csp-view-pcpp}
    Given a $(q,\delta,\varepsilon,f,g)$-PPCPP verifier $A$ for a CSP $G = (V,E,\Sigma,\{\Pi_e\}_{e\in E})$ with $\Sigma = \mathbb F^d$, we define a CSP instance $G' = (V', E', \Sigma', \{\Pi_{e}'\}_{e\in E'})$, where $V'=V_1'\dot\cup V_2'\dot\cup V_3'$ and $\Sigma'=(\Fbb^d)^q$, by the following steps:
    \begin{itemize}
        \item First, for $i=1,2$, we treat each position of $\pi_i$ as a single variable in $V_i'$ with alphabet $\Fbb^d$.
        Note that $|V_1'|=|\Fbb|^k$ and $|V_2'|\le f(k)$.
        \item Then, for each randomness $r\in [R_A]$, let $S_r$ be the set of query positions over $\pi_1\circ \pi_2$ under randomness $r$; and we add a supernode $z_r$ to $V_3'$ whose alphabet is $(\Fbb^d)^{|S_r|}$, i.e., all possible configurations of the query result.
        Note that $|V_3'|\le g(k)$.
        \item Finally, we add constraints between $z_r$ and every query position $i\in S_r$. The constraint checks whether $z_r$ is an accepting configuration, and the assignment of the position $i$ is consistent with the assignment of $z_r$. 
    \end{itemize}
\end{definition}

By construction, we can see that the completeness and soundness are preserved up to a factor of $q$ under this reduction, where the loss comes from the construction where we split $q$ queries into $q$ consistency checks.
In addition, since $|\pi_1\circ\pi_2|\le|\Fbb|^k+f(k)$, $R_A\le g(k)$, and the list of queries made by $A$ can be generated in time $h(k)\cdot|G|^{O(1)}$, the reduction from $G$ to $G'$ is a FPT reduction.

\begin{fact}\label{fct:ppcpp}
The reduction described in \Cref{def:csp-view-pcpp} is an FPT reduction.
Recall that $k=|V|$ is the parameter of $G$ and $\Sigma=\Fbb^d$ is the alphabet of $G$.
We have the following properties for $G'$: 
    \begin{itemize}
        \item \textsc{Alphabet.}
        The alphabet of $G'$ is $\Sigma'=\Fbb^{d\cdot q}$.
        \item \textsc{Parameter Blowup.}
        The parameter of $G'$ is $|V'|\le|\Fbb|^k+f(k)+g(k)$.
        \item \textsc{Completeness.} 
        For every solution $\sigma$ of $G$, there exists a solution $\sigma'$ of $G'$ assigning $\pwh(\sigma)$ to $V_1'$.
        \item \textsc{Soundness.} 
        For any assignment $\sigma'$ satisfying $1-\frac{\varepsilon}{q}$ fraction of the constraints in $G'$, there exists a solution $\sigma$ of $G$ such that $\Delta(\sigma'(V_1'), \pwh(\sigma))\le \delta$.
    \end{itemize}    
\end{fact}
\section{Proof of The Main Theorem}\label{sec:structure}
In this section, we prove the following quantitative version of our main theorem  (\Cref{thm:main}). To depict a clear picture, we will treat some technical constructions as black-boxes and relegate their proofs in subsequent sections. 

\begin{theorem}\label{thm:main_formal}
Assume ETH is true.
No algorithm can decide \textsc{$\frac1{9600}$-Gap $k$-CSP} within runtime $f(k)\cdot n^{o(\sqrt{\log\log k})}$ for any computable function $f$.
\end{theorem}

As a byproduct of the quantitative analysis, we also have the following PCP-style theorem, which can be viewed as a parameterized PCP theorem.
\begin{theorem}\label{thm:main_formal_pcp}
For any integer $k\ge1$, {\sc 3SAT} has a PCP verifier which
\begin{itemize}
\item can be constructed in time $f(k)\cdot|\Sigma|^{O(1)}$ for some computable function $f$,
\item makes two queries on a proof of length $2^{2^{O(k^2)}}$ and alphabet size $|\Sigma|=2^{O(n/k)}$,
\item has perfect completeness and soundness $1-\frac1{9600}$.
\end{itemize}
\end{theorem}

Our proof relies on an intermediate structured CSP, termed \emph{Vector-Valued CSPs (\vcsp{} for short)}.

\begin{figure}[ht]
    \centering
    \includegraphics[width=\textwidth]{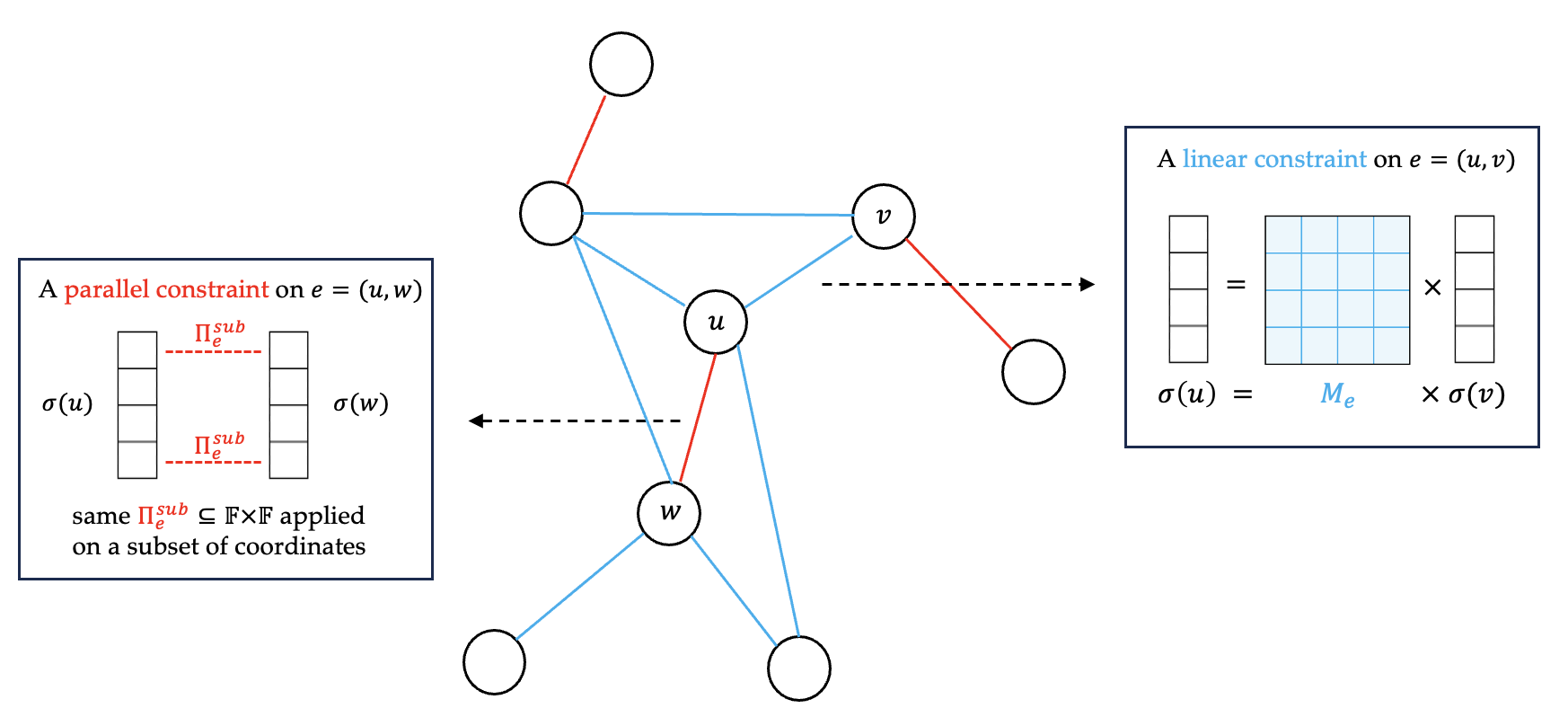}
    \caption{An example of a vector-valued CSP.}
    \label{fig:vcsp-1}
\end{figure}

\begin{definition}[Vector-Valued CSP]\label{def:vcsp}
A CSP instance $G = (V,E,\Sigma,\{\Pi_{e}\}_{e\in E})$ is a \vcsp{} if the following additional properties hold.
\begin{itemize}
    \item $\Sigma=\Fbb^d$ is a $d$-dimensional vector space over a finite field $\mathbb F$ with characteristic $2$.
    \item For each constraint $e=\cbra{u,v}\in E$ where $u = (u_1,u_2,\dots,u_d)$ and $v = (v_1,v_2,\ldots,v_d)$, the constraint validity function $\Pi_e$ is classified as one of the following forms in order\footnote{A constraint can be both linear and parallel (e.g., equality constraint). In this case, we classify it as linear instead of parallel, consistent with the order defined here.}:
    \begin{itemize}
        \item \textsc{Linear.} There exists a matrix\footnote{In the instance reduced from \textsc{3SAT}, $M_e$ is always a permutation matrix.} $M_e\in \mathbb F^{d\times d}$ such that 
        $$
        \Pi_e(u,v) = \indicator_{u = M_ev}.
        $$
        \item \textsc{Parallel.} There exists a sub-constraint $\Pi_e^{sub}: \mathbb F\times \mathbb F\to \bin$ and a subset of coordinates $Q_e\subseteq[d]$ such that $\Pi_e$ checks $\Pi_e^{sub}$ for every coordinate in $Q_e$, i.e., 
        $$
        \Pi_e(u,v) = \bigwedge_{i\in Q_e} \Pi_e^{sub}(u_i,v_i).
        $$
    \end{itemize}
    \item Each variable is related to at most one parallel constraint.
\end{itemize}
\end{definition}

We refer to \Cref{fig:vcsp-1} as an illustration of \vcsp{}.

Our reduction is accomplished by combining two separate sub-reductions.
First, in \Cref{sec:tensor-csp-construction}, we provide a reduction from {\sc 3SAT} to \vcsp{}s.
Second, in \Cref{sec:producing-gaps}, we provide another reduction from \vcsp{}s to parameterized CSPs of constant gap.
Finally, in \Cref{sec:combine-reduction}, we show how to combine the two reductions above to prove \Cref{thm:main_formal} and \Cref{thm:main_formal_pcp}.

\subsection{Reduction I: From {\sc 3SAT} to Vector-Valued CSPs}\label{sec:tensor-csp-construction}
In this step, we reduce {\sc 3SAT} to \vcsp{}s. 
By \Cref{hypo:eth}, we may assume {\sc 3SAT} has some additional structure.

\begin{theorem}[Proved in \Cref{sec:3sat-to-csp}]\label{thm:3satCSP}
There is a reduction algorithm such that the following holds.
For any positive integer $\ell$ and given as input a SAT formula $\varphi$ of $n$ variables and $m$ clauses, where each variable is contained in at most four clauses and each clause contains exactly three distinct variables, the reduction algorithm produces a \vcsp{} instance $G = (V,E,\Sigma,\{\Pi_e\}_{e\in E})$ where:
\begin{enumerate}[label=\textbf{(S\arabic*)}]
    \item\label{itm:thm:3sat2csp_1} \textsc{Variables and Constraints.} $|V|=48\ell^2$ and $|E|=72\ell^2$.
    \item\label{itm:thm:3sat2csp_2} \textsc{Runtime.} The reduction runs in time $\ell^{O(1)}\cdot2^{O(n/\ell)}$.
    \item\label{itm:thm:3sat2csp_3} \textsc{Alphabet.} $\Sigma = \Fbb_8^d$ where $d=\max\cbra{\ceilbra{m/\ell},\ceilbra{n/\ell}}$.
    \item\label{itm:thm:3sat2csp_5} \textsc{Completeness and Soundness.} $G$ is satisfiable iff $\varphi$ is satisfiable.
\end{enumerate}
\end{theorem}

\subsection{Reduction II: From Vector-Valued CSPs to Gap CSPs}\label{sec:producing-gaps}

Now we present our gap-producing reduction from \vcsp{}s to instances of $\gapkcsp$.

\begin{theorem}\label{prop:reduction}
Fix an absolute constant $\eps^* = \frac{1}{9600}$.
There is a reduction algorithm such that the following holds.
Given as input a \vcsp{} instance $G = (V,E,\Sigma=\Fbb^d,\{\Pi_e\}_{e\in E})$ where 
\begin{itemize}
    \item $k=|V|$ is the parameter of $G$,
    \item $|\Fbb| = 2^{t}\le h(k)$ for some computable function $h$,
    \item $|E|\le m(k)$\footnote{Note that we allow multiple constraints between a same pair of variables. Hence in general $|E|$ may not be bounded by a function of $k$.} for some computable function $m$ such that $m(k)\ge1$,
\end{itemize}  
the reduction algorithm produces a CSP instance $G^*=(V^*,E^*,\Sigma^*=\Fbb^{4d},\{\Pi_e^*\}_{e\in E^*})$ where:
\begin{itemize}
    \item \textsc{FPT reduction.} The reduction from $G$ to $G^*$ is an FPT reduction.
    \item \textsc{Parameter blowup.} The parameter of $G^*$ is $|V^*|\le2^{2^k\cdot m(k)\cdot|\Fbb|^{O(1)}}$.
    \item \textsc{Completeness.} If $G$ is satisfiable, then $G^*$ is satisfiable.
    \item \textsc{Soundness.} If $G$ is not satisfiable, then $\val(G^*)<1-\varepsilon^*$.
\end{itemize}
\end{theorem}

Below, we present our reduction and proof for \Cref{prop:reduction}. Fix a \vcsp{} instance $G = (V,E,\Sigma,\{\Pi_e\}_{e\in E})$ satisfying the conditions in \Cref{prop:reduction}. 
Our reduction is achieved in three steps.

\subsubsection{Step a: Instance Splitting}

Recall that $G$ has two kinds of constraints: linear and parallel constraints. In this step,  we partition the constraint set $E$ into two parts $E_L\dot\cup E_P$, where $E_L$ and $E_P$ consist of all linear and parallel constraints of $E$, and define $G_L = (V, E_L, \Sigma, \{\Pi_e\}_{e\in E_L})$ and $G_P = (V, E_P, \Sigma, \{\Pi_e\}_{e\in E_P})$ as the sub-CSP instance where the constraint set is $E_L$ and $E_P$, respectively.
Note that $G_L$ and $G_P$ are still \vcsp{}s with the same parameter $k=|V|$. Furthermore, we have the simple observation as follows.
\begin{fact}\label{fct:instance_splitting}
For every assignment $\sigma$ over $V$, $\sigma$ is a solution of $G$ if and only if it is the solution of both $G_L$ and $G_P$.
\end{fact}

\subsubsection{Step b: Designing Parallel PCPPs for Sub-Instances}

In this step, we construct PPCPP verifiers $A_L$ and $A_P$ in FPT time to test whether all constraints in $G_L$ and $G_P$ are satisfied, respectively. We first handle parallel constraints and obtain $A_P$.
\begin{proposition}[PPCPP for Parallel Constraints. Proved in \Cref{sec:parallel-pcpp}]\label{prop:pcpp-parallel-constraint}
Let $h$ be a computable function.
Let $G$ be a \vcsp{} instance with $k$ variables where (1) the alphabet is $\Fbb^d$ and $|\Fbb|=2^t\le h(k)$, and (2) all constraints are parallel constraints.
Then for every $\eps\in(0,\frac{1}{800})$, there is a $(4,48\varepsilon,\varepsilon,f(k)=2^{2^k\cdot |\Fbb|^{O(1)}},g(k)=2^{2^k\cdot |\Fbb|^{O(1)}})$-PPCPP verifier for $G$, where $f(k)$ is the length of the auxiliary proof, and $g(k)$ is the number of random choices.
\end{proposition}

Recall that the alphabet of $G$ is $\Fbb^d$ where $|\Fbb| = 2^t$, and $G_P$ consists of parallel constraints of $G$ only. Thus, by plugging $\eps=\frac{1}{1200}$ into the proposition above, we can obtain a $(q_P=4,\delta_P=\frac{1}{25},\varepsilon_P=\frac{1}{1200},f_P(k)=2^{2^k\cdot |\Fbb|^{O(1)}},g_P(k)=2^{2^k\cdot |\Fbb|^{O(1)}})$-PPCPP verifier $A_P$ for $G_P$. Now, we turn to linear constraints and obtain $A_L$.

\begin{proposition}[PPCPP for Linear Constraints. Proved in \Cref{sec:linear-pcpp}]\label{prop:pcpp-linear-constraint}
Let $h$ and $m$ be two computable functions.
Let $G$ be a \vcsp{} instance with $k$ variables where (1) the alphabet is $\Fbb^d$ and $|\Fbb|\le h(k)$, (2) all constraints are linear constraints, and (3) there are at most $m(k)$ constraints.
Then for every $\varepsilon\in \left(0, \frac1{400}\right)$, there is a $(4,24\varepsilon,\varepsilon,f(k)=|\Fbb|^{k\cdot m(k)},g(k)=|\Fbb|^{8k\cdot m(k)})$-PPCPP verifier for $G$.
\end{proposition}

By plugging $\eps=\frac1{600}$ into the proposition above, we can derive a $(q_L=4,\delta_L=\frac1{25},\varepsilon_L=\frac{1}{600},f_L(k)=|\Fbb|^{k\cdot m(k)},g_L(k)=|\Fbb|^{8k\cdot m(k)})$-PPCPP verifier $A_L$ for $G_L$. 

Now, we combine $A_L$ and $A_P$ into a single PPCPP $A$ for the general \vcsp{} $G$ from \Cref{prop:reduction}. In step c, we will convert $A$ into a CSP instance with an inherent gap, completing the proof of \Cref{prop:reduction}.

Here $A$ executes $A_L$ and $A_P$ as in a black-box way where $A$ takes $\pi_1\circ\pi_L\circ\pi_P$ as a proof and with equal probability, $A$ invokes $A_L$ with proof $\pi_1\circ\pi_L$ or invokes $A_P$ with proof $\pi_1\circ\pi_P$.

Intuitively, $\pi_1$ serves as a unified encoding of a solution of $G$ via the parallel Walsh-Hadamard code $\pwh$, and $\pi_L$ and $\pi_P$ are auxiliary proofs to convince $A_L$ and $A_P$ respectively. 
The following proposition shows that $A$ is a PPCPP that efficiently checks all the constraints in $G$.

\begin{proposition}[Combined PCPP]\label{prop:combined}
Given a \vcsp{} instance $G$ satisfying the preconditions in \Cref{prop:reduction}, the verifier $A$ described above is a $(q=4,\delta=\frac{1}{25},\eps=\frac{1}{2400},f(k) = 2^{2^k\cdot m(k)\cdot|\Fbb|^{O(1)}},g(k) = 2^{2^k\cdot m(k)\cdot|\Fbb|^{O(1)}})$-PPCPP verifier for $G$.
\end{proposition}
\begin{proof}
Since $A$ invokes either $A_L$ or $A_P$, the number of queries is the maximum of $q_L=4$ and $q_P=4$. 
The length of the auxiliary proof is 
$$
|\pi_L \circ \pi_P|\le f_L(k) + f_P(k)=|\Fbb|^{k\cdot m(k)}+2^{2^k\cdot|\Fbb|^{O(1)}}\le2^{2^k\cdot m(k)\cdot|\Fbb|^{O(1)}}=f(k).
$$
For alignment, we pad the randomness of $A_P$ and $A_L$ to ensure that they have same amount 
$$
R\le g_P(k)g_L(k)=|\Fbb|^{8k\cdot m(k)}\cdot2^{2^k\cdot|\Fbb|^{O(1)}}\le2^{2^k\cdot m(k)\cdot|\Fbb|^{O(1)}}
$$
of uniform choices. 
Thus the total number of uniform choices of $A$ is $2R\le g(k)$. 
Finally, we analyze the completeness and soundness.

\paragraph*{Completeness.} 
Let $\sigma$ be a solution of $G$. We set $\pi_1=\pwh(\sigma)$.
By \Cref{fct:instance_splitting}, $\sigma$ is also a solution of $G_L$ and $G_P$. As a result, by the definition of PPCPP~(\Cref{def:pcpp}), there exists $\pi_L$ and $\pi_P$ such that $A_L$ and $A_P$ always accept $\pi_1\circ \pi_L$ and $\pi_1\circ \pi_P$ respectively. Thus, $A$ always accepts $\pi_1\circ \pi_L\circ \pi_P$.

\paragraph*{Soundness.} Assume $A$ accepts $\pi_1\circ \pi_L\circ \pi_P$ with probability at least $1-\eps$.
By construction, $A_L$ accepts $\pi_1\circ \pi_L$ with probability at least $1-2\eps\ge 1-\varepsilon_L$. Thus by the definition of PPCPP (\Cref{def:pcpp}), there exists a solution $\sigma_L$ of $G_L$ such that $\Delta(\pi_1, \pwh(\sigma_L))\le \delta_L\le \delta$. 
Similarly for $A_P$, there exists a solution $\sigma_P$ of $G_P$ such that $\Delta(\pi_1, \pwh(\sigma_P))\le \delta$. 
Thus 
$$
\Delta(\pwh(\sigma_P), \pwh(\sigma_L))\le2\delta<\frac12.
$$
Recall from \Cref{sec:hadamard_code} that the relative distance of $\pwh$ is at least $\frac12$. Hence we must have that $\sigma_L = \sigma_P$, which means $\sigma_L=\sigma_P$ is a solution of $G$ such that $\Delta(\pi_1, \pwh(\sigma_L))\le \delta$.
\end{proof}

\subsubsection{Step c: Reducing Parallel PCPPs to Gap CSPs}

Finally, we complete the proof of \Cref{prop:reduction} by converting the verifier $A$ into a constant-gap parameterized CSP $G^*$.

\begin{proof}[Proof of \Cref{prop:reduction}]
Set $G^* = (V^*, E^*, \Sigma^*, \{\Pi_{e}^*\}_{e\in E^*})$ to be the CSP instance obtained by applying the reduction in \Cref{def:csp-view-pcpp} on the verifier $A$.
Then the claimed runtime of the reduction, as well as the alphabet size,  completeness and soundness of $G^*$, follow immediately from combining \Cref{prop:combined} and \Cref{fct:ppcpp}.
Here we simply note the parameter blowup: $$|V^*|\le |\Fbb|^k + f(k)+g(k) =2^{2^k\cdot m(k)\cdot|\Fbb|^{O(1)}} \ .  \qedhere$$
\end{proof}

\subsection{Putting Everything Together}\label{sec:combine-reduction}
In this part, we combine~\Cref{thm:3satCSP} and \Cref{prop:reduction} to prove \Cref{thm:main_formal} and \Cref{thm:main_formal_pcp}.

\begin{proof}[Proof of \Cref{thm:main_formal}]
Assuming ETH (\Cref{hypo:eth}), there is no $2^{o(n)}$ algorithm for deciding {\sc 3SAT} formula $\varphi$ where each variable is contained in at most four clauses and each clause contains exactly three distinct variables. 
For any such a formula $\varphi$, we show how to reduce $\varphi$ to a parameterized CSP instance $G^*$ with a constant inherent gap.

Let $\ell$ be a parameter to be chosen later.
We first invoke~\Cref{thm:3satCSP} and obtain a \vcsp{} instance $G$ in $\ell^{O(1)}\cdot2^{O(n/\ell)}$ time where:
\begin{itemize}
    \item There are $k=48\ell^2$ variables and their alphabet is $\Fbb_8^d$ with $d=O(n/\ell)$.
    \item $G$ is satisfiable iff $\varphi$ is satisfiable. 
\end{itemize}
Note that the size of $G$ is $|G|=\ell^{O(1)}\cdot2^{O(n/\ell)}$.

Then we apply \Cref{prop:reduction} on $G$ with parameters $h(k)=8$ and $m(k)=2k$.
After this reduction, we obtain a CSP instance $G^*=(V^*,E^*,\Sigma^*,\{\Pi_e^*\}_{e\in E^*})$ such that:
\begin{itemize}
    \item The runtime and instance size $N:=|G^*|$ of the reduction are bounded by $r_1(k)\cdot|G|^{O(1)}=r_2(\ell)\cdot2^{O(n/\ell)}$ for some computable functions $r_1,r_2$.
    \item The parameter of $G^*$ is $K=|V^*|\le2^{2^{O(k)}}=2^{2^{O(\ell^2)}}$.
    \item If $G$ is satisfiable, then $G^*$ is satisfiable.
    \item If $G$ is not satisfiable, then $\val(G^*)<1 - \frac{1}{9600}$.
\end{itemize}

In the paramterized complexity theory, we treat the parameter $K=|G^*|$ of $G^*$ as a superconstant that is much smaller than the size $N=|G^*|$ of $G^*$.
This means that the initial parameter $\ell$ is also a superconstant that is much smaller than $n$.
Therefore the total reduction time $\ell^{O(1)}\cdot2^{O(n/\ell)}+r_2(\ell)\cdot2^{O(n/\ell)}$ is still $2^{o(n)}$.

Since the size of $G^*$ is $N\le r_2(\ell)\cdot2^{O(n/\ell)}$, ETH (\Cref{hypo:eth}) rules out any algorithm with runtime $(N/r_2(\ell))^{o(\ell)}$ to decide $\left(\frac{1}{9600}\right)$-{\sc Gap} $K$-{\sc CSP}.
Since $K\le2^{2^{O(\ell^2)}}$ and by the same parameterized complexity theoretic perspective, this encompasses algorithms with runtime $f(K)\cdot N^{o(\sqrt{\log\log K})}$ for any computable function $f$.
\end{proof}

The above quantitative analysis readily gives the PCP-style statement (\Cref{thm:main_formal_pcp}).
\begin{proof}[Proof of \Cref{thm:main_formal_pcp}]
Given a {\sc 3SAT} formula of size $n$, we apply the sparsification lemma \cite{IPZ01} and Tovy's reduction \cite{Tov84} to obtain a {\sc 3SAT} instance $\varphi$ on $n$ Boolean variables where each variable is contained in at most four clauses and each clauses contains exactly three distinct variables.
In addition, this reduction runs in $n^{O(1)}$ time and preserves the satisfiability of the original {\sc 3SAT} formula.

Let $\ell\ge1$ be a parameter.
Then we apply the analysis of \Cref{thm:main_formal} on $\varphi$ to obtain a CSP instance $G^*=(V^*,E^*,\Sigma^*,\{\Pi_e^*\}_{e\in E^*})$ with $|V^*|=K$ variables and alphabet size $|\Sigma^*|\le2^{O(n/\ell)}$ in time $r_2(\ell)\cdot2^{O(n/\ell)}$, where $K=2^{2^{O(\ell^2)}}$ and $r_2$ is some computable function.
In addition, if $\varphi$ is satisfiable, then $\val(G^*)=1$; otherwise $\val(G^*)<1-\frac1{9600}$.

Now a PCP verifier takes a proof $\pi\colon V^*\to\Sigma^*$, viewed as a string of length $K$ and alphabet $\Sigma^*$, and picks a uniform random constraint $e\in E^*$ to check.
Note that the runtime of this verifier is bounded by the size $|G^*|\le r_2(\ell)\cdot2^{O(n/\ell)}$ of $G^*$.
If the original {\sc 3SAT} formula is satisfiable, then $\varphi$ is satisfiable and thus there exists a proof that always passes the check.
Otherwise, by the soundness guarantee of $G^*$, any proof will violate at least $\frac1{9600}$ fraction of the constraints in $E^*$, which implies the soundness gap of the PCP verifier.
Setting $\ell=k$ and $f(k)=r_2(k)$ completes the proof of \Cref{thm:main_formal_pcp}.
\end{proof}
\section{From {\sc 3SAT} to Vector-Valued CSP}\label{sec:3sat-to-csp}

This section is devoted to the proof of \Cref{thm:3satCSP} which shows how to obtain a \vcsp{} from a {\sc 3SAT} instance.

Before going to the details, we mark the high level picture of the reduction as follows:
\begin{enumerate}
    \item First, we divide the clauses and variables of the {\sc 3SAT} instance $\varphi$ into $k$ parts, and build a vector-valued variable (in the following, we will denote them as ``vertices'' in order to distinguish them from the variables in $\varphi$) for each part of clauses and each part of variables. Then we apply tests for checking the consistency between clauses and variables.
    \item Next, we appropriately duplicate each vertex into several copies. Then we split the constraints and spread them out to different copies of vertices, such that
    \begin{itemize}
        \item each vertex is related to at most one constraint;
        \item the sub-constraints inside each constraint form a matching on the $2d$ coordinates of the two endpoints.
    \end{itemize}
    \item Finally, given the properties above, we can rearrange the $d$ coordinates of each vertex according to its only constraint, to make the sub-constraints parallel. Furthermore, we add a cycle of constraints on the copies of each vertex, forcing them to take the same value (before rearranging the coordinates). Such constraints are then permuted equalities, which are in turn special linear constraints.
\end{enumerate}

We refer to \Cref{fig:3sat2vcsp-1} for an informal illustration of the above process.

\begin{figure}[ht]
    \centering
    \includegraphics[width=\textwidth]{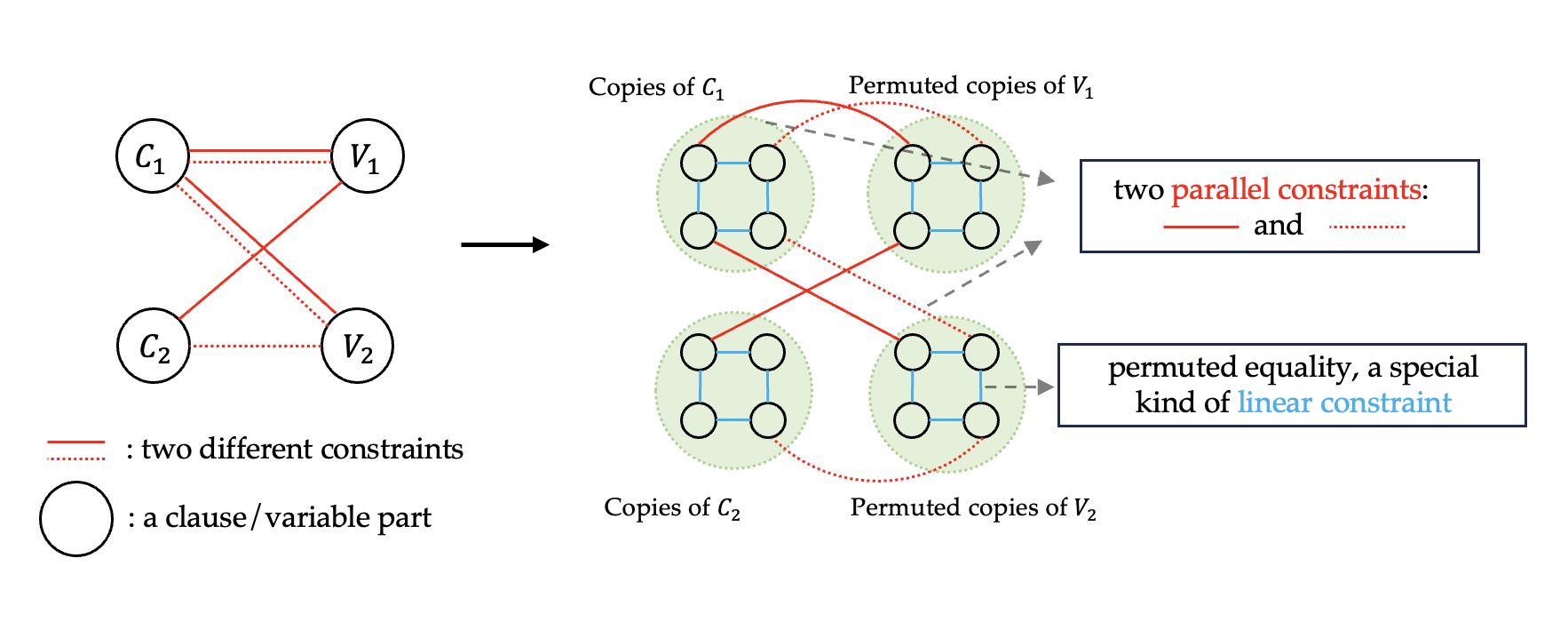}
    \caption{An example showing the vertex duplicating and constraint splitting steps.}
    \label{fig:3sat2vcsp-1}
\end{figure}

\begin{theorem*}[\Cref{thm:3satCSP} Restated]
There is a reduction algorithm such that the following holds.
For any positive integer $\ell$ and given as input a SAT formula $\varphi$ of $n$ variables and $m$ clauses, where each variable is contained in at most four clauses and each clause contains exactly three distinct variables, the reduction algorithm produces a \vcsp{} instance $G = (V,E,\Sigma,\{\Pi_e\}_{e\in E})$ in $\ell^{O(1)}\cdot2^{O(n/\ell)}$ time, where:
\begin{enumerate}[label=\textbf{(S\arabic*)}]
    \item \textsc{Variables and Constraints.} $|V|=48\ell^2$ and $|E|=72\ell^2$.
    \item \textsc{Runtime.} The reduction runs in time $\ell^{O(1)}\cdot2^{O(n/\ell)}$.
    \item \textsc{Alphabet.} $\Sigma = \Fbb_8^d$ where $d=\max\cbra{\ceilbra{m/\ell},\ceilbra{n/\ell}}$.
    \item \textsc{Completeness and Soundness.} $G$ is satisfiable iff $\varphi$ is satisfiable.
\end{enumerate}
\end{theorem*}

Fix $\varphi$ and $\ell$ from \Cref{thm:3satCSP}. 
For each variable, we fix an arbitrary order of its (at most four) appearances in clauses. The order is used to construct parallel constraints.

We partition $m$ clauses of $\varphi$ into $\Ccal_1,\ldots,\Ccal_\ell$ where each $\Ccal_i$ contains at most $\ceilbra{m/\ell}$ clauses.
Similarly we partition $n$ variables of $\varphi$ into $\Vcal_1,\ldots,\Vcal_\ell$ where each $\Vcal_i$ contains at most $\ceilbra{n/\ell}$ variables.
For each clause $C\in\Ccal_1\dot\cup\cdots\dot\cup\Ccal_\ell$, we identify $\Fbb_8=\bin^3$ as the set of partial assignments to the clause $C$, where $(0,0,0)\in\Fbb_8$ is the only unsatisfying assignment of $C$. For each variable $x\in \Vcal_1\dot\cup\cdots\dot\cup\Vcal_\ell$, we also treat its assignment $\in \{0,1\}$ as an element of $\Fbb_8$. 

Formally, given a clause $C=y_{i_1}\lor y_{i_2}\lor y_{i_3}$, where each literal $y_{i_j}$ equals variable $x_{i_j}$ or its negation $\neg x_{i_j}$,  every $\tau\in\Fbb_8$, viewed as an element in $\bin^3$, corresponds to a unique assignment by setting $y_{i_j}=\tau(j)$ for $j\in[3]$, which in turn assigns the value of $x_{i_j}$.


Now we define six tests as sub-constraints to be used later.
For $j\in[3]$ and $b\in\bin$, we define $\Pi_{j,b}\colon\Fbb_8\times\Fbb_8\to\bin$ by
$$
\Pi_{j,b}(\tau,c)=\indicator_{c\in\bin}\cdot\indicator_{\tau\neq(0,0,0)}\cdot\indicator_{\tau(j)=c\oplus b},
$$
Intuitively, these constraints checks that: (1) the variable assignment is binary, (2) the clause assignment is satisfying, and (3) the clause assignment and variable assignment are consistent.

\paragraph{Vertices and Alphabets.} We first define the vertices and the alphabet of $G$. In detail, for each $p\in[\ell],q\in[\ell],j\in[3],s\in [4],b\in\bin$, we put into $V$ a vertex $z_{p,q,j,s,b}$ with alphabet $\Fbb_8^{|\Ccal_p|}$, and a vertex $w_{p,q,j,s,b}$ with alphabet $\Fbb_8^{|\Vcal_q|}$. Intuitively, each vector entry of $z_{p,q,j,s,b}$ corresponds to the assignment of a clause $\in \Ccal_p$, and each vector entry of $w_{p,q,j,s,b}$ corresponds to the assignment of a variable $\in \Vcal_q$. Thus, we index entries of $z_{p,q,j,s,b}$ by clauses in $\Ccal_p$ and entries of $w_{p,q,j,s,b}$ by variables in $\Vcal_q$. 
Since $d= \max\{\lceil m/\ell\rceil,\lceil n/\ell\rceil\}=\max\{|\Ccal_p|, |\Vcal_q|\}$, some entries may be left unused.

At a high level, the vertices $z_{p,q,j,s,b}$ and $w_{p,q,j,s,b}$ are duplicates of assignments to $\Ccal_p$ and $\Vcal_q$ respectively. They will be used to check the test $\Pi_{j,b}$ between clauses in $\Ccal_p$ and variables in $\Vcal_q$ for the $s$-th appearance. Note that since we assume that every variable in $\varphi$ is contained in at most four clauses, we can safely restrict the range of $s$ to be $[4]$.


Now, we can verify that \Cref{itm:thm:3sat2csp_1} holds since in total we have $\ell\cdot \ell\cdot 3\cdot 2\cdot 4\cdot 2=48\ell^2$ vertices in $V$ and each vertex is related to one parallel constraint and two linear constraints. In addition, \Cref{itm:thm:3sat2csp_3} holds since all variables have alphabet $\Fbb_8^d$.

\paragraph{Constraints.} Below, we describe the constraints in the \vcsp{} $G$. At the beginning, we add parallel constraint between $z_{p,q,j,s,b}$ and $w_{p,q,j,s,b}$. 
For simplicity, in this paragraph, we use $\zeta$ to denote a choice of $p\in[\ell],q\in[\ell],j\in[3],s\in[4]$. Below, we enumerate every $\zeta\in[\ell]\times[\ell]\times[3]\times[4]$ and $b\in \bin$. We first define
$$
T_{\zeta,0}=\cbra{(C,x)\in\Ccal_p\times\Vcal_q\colon\text{the $s$-th appearance of variable $x$ is the $j$-th literal in clause $C$ as $x$}}
$$
and
$$
T_{\zeta,1}=\cbra{(C,x)\in\Ccal_p\times\Vcal_q\colon\text{the $s$-th appearance of variable $x$ is the $j$-th literal in clause $C$ as $\neg x$}}.
$$

Then, for every $(C,x)\in T_{\zeta,b}$, we put a sub-constraint $\Pi_{j,b}$ ($j$ is encapsulated in $\zeta=(p,q,j,s)$) between the $C$-th entry of $z_{\zeta,b}$ and the $x$-th entry of $w_{\zeta,b}$, which checks whether the assignment of literals in $C$ is consistent with the assignment of $x$.
Observe that between entries of $z_{\zeta,b}$ and $w_{\zeta,b}$, we only put the sub-constraint $\Pi_{j,b}$. In addition, $T_{\zeta,b}$ forms a (not necessarily perfect) matching over $\Ccal_p\times\Vcal_q\subseteq[d]\times[d]$ as any two distinct $(C,x),(C',x')\in T_{\zeta,b}$ satisfy $C\neq C'$ and $x\neq x'$. Thus, we can rearrange entries of $w_{\zeta,b}$  so that the sub-constraints between $z_{\zeta,b}$ and $w_{\zeta,b}$ is parallel. We use $\kappa_{\zeta,b}\colon[d]\to[d]$ to denote the permutation applied in the rearrangement, i.e., $\kappa_{\zeta,b}(C)=x$ for all $(C,x)\in T_{\zeta,b}$. Specifically, $w_{\zeta,b}$ is rearranged in such a way that its new $C$-th entry takes the value of its old $\kappa_{\zeta,b}(C)$-th entry.

Finally, we remark that each variable $w_{\zeta,b}$ only need to be rearranged once according to $\kappa_{\zeta,b}$. Thus, the constraint between $z_{\zeta,b}$ and $w_{\zeta,b}$ is well-defined. From the construction above, we obtain a parallel constraint $e:=\cbra{z_{\zeta,b},w_{\zeta,b}}\in E$ with the sub-constraint $\Pi_e^{sub} = \Pi_{j,b}$ and $Q_e=\cbra{C\colon(C,x)\in T_{\zeta,b}}$.

After adding constraints for ``clause-variable'' consistency, we need to further establish consistency check to ensure $z_{p,\cdot,\cdot,\cdot,\cdot}$ corresponds to the same assignment over $\Ccal_p$. Similarly, we also need a consistency check to ensure that $w_{\cdot,q,\cdot,\cdot,\cdot}$ corresponds to the same assignment for $\Vcal_q$. Thus, we need constraints as follows.
\begin{itemize}
\item For each $p\in[\ell]$, we connect $\{z_{p,q,j,s,b}: q\in[\ell],j\in[3],s\in[4],b\in\bin\}$ in the constraint graph $G$ by an arbitrary cycle (denoted by the cycle of $\Ccal_p$), where every constraint in this cycle is a linear constraint. In detail, for every two vertices $\hat z=z_{p,q,j,s,b}$ and $\tilde z=z_{p,q',j',b',s'}$ connected in the cycle of $\Ccal_p$, we impose the linear constraint that $\indicator_{\hat z = \tilde z}$.

\item Similarly, for each $q\in [\ell]$, we also connect $\{w_{p,q,j,s,b}: p\in[\ell],j\in[3],s\in[4],b\in\bin\}$ in the constraint graph $G$ by an arbitrary cycle (denoted by the cycle of $\Vcal_q$). Also, every constraint in this cycle is a linear constraint. Note that we have rearragned $w_{p,q,j,s,b}$ by the permutation $\kappa_{p,q,j,s,b}$ to ensure the constraint between $z$ and $w$ is parallel. Thus, we add the permutated equality between two connected vertices $\hat w$ and $\tilde w$ in the cycle. In detail, we impose the linear constraint that $\indicator_{\hat w=M_{\hat w,\tilde w}\tilde w}$ where $M_{\hat w,\tilde w}\in\bin^{d\times d}$ is the permutation matrix of the permutation $\kappa_{\hat w}\circ\kappa_{\tilde w}^{-1}$.
\end{itemize}

To see that the instance is indeed a \vcsp{}, we observe that every $z_{p,q,j,s,b}$ and $w_{p,q,j,s,b}$ is related to at most one parallel constraint.


Since every variable of $\phi$ is contained in at most four clauses and every clause contains three distinct variables, we have $m\le4n/3$.
Hence $|\Sigma|=8^d=2^{O(n/\ell)}$ and the construction of $G$ above can be done in time $\ell^{O(1)}\cdot2^{O(n/\ell)}$, showing \Cref{itm:thm:3sat2csp_2}.

Finally, we establish the completeness and soundness to verify \Cref{itm:thm:3sat2csp_5}, which almost writes itself given the construction above.

\paragraph*{Completeness.}
Assume $\varphi$ is satisfiable by an assignment $\sigma$ to the variables.
This implies an assignment $\tau$ to the literals in clauses.
Then for each entry $C$ of $z_{p,q,j,s,b}$, we assign it as $\tau(C)$, which is among $\bin^3\setminus\cbra{(0,0,0)}$ as $\sigma$ is a satisfying assignment.
For each entry $x$ of $w_{p,q,j,s,b}$, we assign it as $\sigma(x)$.

It is easy to see that the linear constraints in $G$ are all satisfied, since those are simply checking equality of the assignments.
For parallel constraints, we observe that each sub-constraint $\Pi_{j,b}$ between entry $C$ of $z_{p,q,j,s,b}$ and entry $x$ of $w_{p,q,j,s,b}$ checks whether $x$ is assigned to be consistent with its appearance in $C$, where $b$ indicates if $x$ appears as literal $x$ or $\neg x$, and $j$ represents the location of this literal in $C$.
Since our assignments of these vertices are based on the assignment $\sigma$, it naturally passes all the tests, which finishes the proof of completeness.

\paragraph*{Soundness.}
Assume $G$ is satisfiable.
By the linear constraints among $z$'s in $G$, we obtain an assignment $\tau$ to the clauses of $\varphi$ that satisfies each clause, and by the linear constraints among $w$'s in $G$, we obtain an assignment $\sigma$ to each variable.
Now it remains to show that $\tau$ is consistent with $\sigma$, implying that $\sigma$ really corresponds to a solution of $\varphi$.

Assume towards contradiction that the value $\tau$ assigns variable $x\in\Vcal_q$ in clause $C \in \mathcal C_p$ is different from $\sigma(x)$. Assume $C$ is the $s$-th appearance of variable $x$ and $x$ is at its location $j \in [3]$. In addition, let $b \in \{0,1\}$ indicate whether $x$ appears as literal $\neg x$ in $C$. Then we have a sub-constraint $\Pi_{j,b}$ between the entry $C$ of $z_{p,q,j,b,s}$ and the entry $x$ (before rearrangement) of $w_{p,q,j,b,s}$. This is a contradiction since the test $\Pi_{j,b}$ will force $\tau(C)$ to assign variable $x$ the same value as $\sigma(x)$, which completes the proof of soundness.

\section{Parallel PCPPs for Vector-Valued CSPs with Parallel Constraints}\label{sec:parallel-pcpp}

This section is devoted to proving \Cref{prop:pcpp-parallel-constraint}, which is restated as follows.
\begin{proposition*}[\Cref{prop:pcpp-parallel-constraint} Restated]
Let $h$ be a computable function.
Let $G$ be a \vcsp{} instance with $k$ variables where (1) the alphabet is $\Fbb^d$ and $|\Fbb|=2^t\le h(k)$, and (2) all constraints are parallel constraints.
Then for every $\eps\in(0,\frac{1}{800})$, there is a $(4,48\varepsilon,\varepsilon,f(k)=2^{2^k\cdot|\Fbb|^{O(1)}},g(k)=2^{2^k\cdot|\Fbb|^{O(1)}})$-PPCPP verifier for $G$.
\end{proposition*}

The construction of PPCPP in this section is a generalization of an assignment tester used in the proof of the classic exponential length PCP showing result $\mathsf{NP} \subseteq \mathsf{PCP}[\poly(n),O(1)]$~\cite{arora1998proof}.

There, we first convert a Boolean circuit to a quadratic equation system (known as the \textsc{Quadeq} problem) such that they share the same satisfiability, then we use the Walsh-Hadamard code to encode an assignment to the quadratic equation system and certify its satisfiability.

Our analysis relies on the famous random subsum principle~(see e.g., \cite{arora2009computational}), demonstrated as follows.
\begin{lemma}[Random Subsum Principle]\label{lem:rsp}
Given a finite field $\Fbb$ and distinct $M_1,M_2\in\Fbb^{\ell\times\ell'}$ with $\ell,\ell'\ge1$, then $\Pr_{x\in\Fbb^\ell}\sbra{x^\top M_1\neq x^\top M_2}\ge1-\frac1{|\Fbb|}$.
\end{lemma}

\subsection{An Exposition of the {\sc Quadeq} Problem}\label{sec:quadeq}

In the following, we first define the \textsc{Quadeq} problem, and then introduce a PCP verifier for it. 
While this problem and the construction are standard in the literature (see, e.g., \cite{arora2009computational}), we choose to give a brief exposition here for referencing purpose in our actual construction.

\begin{definition}[$\textsc{Quadeq}$]\label{def:quadeq}
    An instance $\Gamma$ of the $\textsc{Quadeq}$ problem consists of $q$ quadratic equations on $c$ binary variables, written concisely as $D_1,\ldots,D_q \in \mathbb F_2^{c \times c}$ and $b_1,\ldots,b_q\in\Fbb_2$.
    The goal of the \textsc{Quadeq} problem is to decide whether there exists a solution $u \in \mathbb F_2^c$ such that $u^\top D_iu=b_i$ holds for all $i\in[q]$.
\end{definition}

The benefit of using the \textsc{Quadeq} problem is that any Boolean circuit satisfiability problem can be efficiently reduced to \textsc{Quadeq} by introducing dummy variables.

\begin{fact}[Folklore, see e.g., \cite{arora2009computational}]\label{fct:circuit_to_quadeq}
Any Boolean circuit\footnote{A Boolean circuit consists of input gates, as well as AND, OR, NOT gates with fan-in (at most) two and fan-out unbounded. Here, we focus on Boolean circuits with a single output gate.} $\Ccal$ with $c$ gates (including input gates) can be converted into a $\textsc{Quadeq}$ instance $\Gamma$ of $c$ variables and $q=O(c)$ equations in $c^{O(1)}$ time such that $\Ccal$ is satisfiable iff $\Gamma$ is satisfiable.

Moreover, there is a one-to-one correspondence between entries of the assignment of $\Gamma$ and gates of $\Ccal$ such that the following holds.
Let $u\in\Fbb_2^c$ be an assignment of $\Gamma$.
Then $u$ is a solution of $\Gamma$ iff $\Ccal$ represents a valid computation and outputs $1$ after assigning values of entries of $u$ to gates of $\Ccal$ by the above correspondence.
\end{fact}

The \textsc{Quadeq} problem also admits an efficient randomized verifier.
Given instance $\Gamma$, the verifier will make at most four queries on a proof $\pi$ and decide whether to accept or reject. 
If $\Gamma$ is satisfiable, then there is a proof that it always accepts; otherwise, it rejects every proof with a constant probability. 

For completeness and simplicity, we recall the standard (non-parallel) Walsh-Hadamard code over $\Fbb_2$, which enjoys the same local testability (\Cref{thm:pwh_test}) and correctability (\Cref{fct:pwh_correct}) as $\pwh$.
\begin{definition}[Walsh-Hadamard Code over $\Fbb_2$]\label{def:wh_code_F2}
The Walsh-Hadamard encoding $\wh_2(a)$ of $a\in\Fbb_2^n$ enumerates linear combinations of entries of $a$. Formally, $\wh_2(a)[b]=a^\top b$ for each $b\in\Fbb_2^n$.
\end{definition}

The proof $\pi$ for $\Gamma$ consists of a length-$2^c$ binary string $\pi_1$ and a length-$2^{c^2}$ binary string $\pi_2$.
Here, $\pi_1$ is supposed to be the Walsh-Hadamard encoding $\wh_2(u)$ of a solution $u \in \mathbb F_2^c$, and $\pi_2$ is supposed to be the Walsh-Hadamard encoding $\wh_2(w)$ of $w=u u^\top\in\Fbb_2^{c\times c}$ where we view matrix $w$ as a length-$c^2$ vector.
Then, the verifier checks one of the following tests with equal probability.

\begin{enumerate}
        \item\label{itm:quadeq_blr} \textsc{Linearity Test.} 
        Perform BLR test (recall from \Cref{sec:hadamard_code}) on $\pi_1$ or $\pi_2$ with equal probability and three queries. 
        
        By \Cref{thm:pwh_test}, if the test passes with high probability, then $\pi_1$ and $\pi_2$ are close to $\wh_2(u)$ and $\wh_2(w)$ of some $u \in \mathbb F_2^c$ and $w \in \mathbb F_2^{c \times c}$ respectively.
        By the local correctability (\Cref{fct:pwh_correct}), we can assume that we have access to $\wh_2(u),\wh_2(w)$ via $\pi_1,\pi_2$.
        \item\label{itm:quadeq_tensor} \textsc{Tensor Test.} 
        Test whether the $w$ equals to $u u^\top$. 
        This is achieved by generating uniformly random vectors $r,r' \in \mathbb F_2^c$ and making four queries, where two queries are used to obtain $\wh_2(u)[r]$ and $\wh_2(u)[r']$, and the other two queries are used to locally correct $\wh_2(w)[r'r^\top]$ via \Cref{fct:pwh_correct}. The test accepts if $\wh_2(w)[r'r^\top]=r^\top wr'$ equals $\wh_2(u)[r]\cdot\wh_2(u)[r']=r^\top(u  u^\top)r'$.

        By applying the random subsum principle (\Cref{lem:rsp}) twice, we know that if the test passes with high probability, then $w=uu^\top$. Now, we can assume that $w=uu^\top$.
        \item\label{itm:quadeq_constraint} \textsc{Constraint Test.} Check whether $u$ is a solution of $\Gamma$, i.e., whether $u^\top D_iu=b_i$ holds for every $i \in [q]$. This is achieved by generating a uniform $H\subseteq[q]$ and making two queries to locally correct $\wh_2(w)\sbra{\sum_{i\in H}D_i}$ via \Cref{fct:pwh_correct}. The test accepts if $\wh_2(w)\sbra{\sum_{i\in H}D_i}=\sum_{i\in H}u^\top D_iu$ equals $\sum_{i\in H}b_i$.
        
        By \Cref{lem:rsp}, if the test passes with high probability, then $u$ is indeed a solution.
\end{enumerate}

\subsection{From Parallel Constraints to Parallel {\sc Quadeq}}\label{sec:parallel_to_quadeq}

We will generalize the above verifier to the parallel setting to prove \Cref{prop:pcpp-parallel-constraint}.
To this end, we first need to convert the parallel constraints into the {\sc Quadeq} form. 
Here, we will have \emph{parallel} {\sc Quadeq} since the alphabet of \vcsp{} is a vector space of $d$ coordinates.

Recall that we are given a \vcsp{} instance $G$ from \Cref{prop:pcpp-parallel-constraint} with $k$ variables and alphabet $\Fbb^d$, where $|\Fbb|=2^t$ and all constraints are parallel constraints.
We use $V=\cbra{x_1,\ldots,x_k}$ to denote the variables in $G$, and use $E=\cbra{e_1,\ldots,e_m}$ to denote the constraints in $G$.
Recall the definition of \vcsp{} (\Cref{def:vcsp}). We know that each variable is related with at most one parallel constraint, which implies $m\le k/2$.
By rearranging, we assume without loss of generality that $e_\ell$ connects $x_{2\ell-1}$ and $x_{2\ell}$ for each $\ell\in[m]$.
We also recall that a parallel constraint $e_\ell$ checks a specific sub-constraint $\Pi_\ell\colon\Fbb \times \Fbb\to\bin$ on all coordinates in $Q_\ell\subseteq[d]$ simultaneously between $x_{2\ell-1}$ and $x_{2\ell}$.

We will need the following additional notations:
\begin{itemize}
\item Let $\chi\colon\Fbb\to\Fbb_2^t$ be a one-to-one map that flattens elements in $\Fbb$ into $t$ bits. The map $\chi$ preserves the addition operator, i.e., i.e, $\chi(a)+\chi(b)=\chi(a+b)$.
\item For each sub-constraint $\Pi_\ell\colon\Fbb\times\Fbb\to\bin$, we define $\bar\Pi_\ell\colon\Fbb_2^t\times\Fbb_2^t\to\bin$ by setting $\bar\Pi_\ell(a,b)=\Pi_e(\chi^{-1}(a),\chi^{-1}(b))$ for all $a,b\in\Fbb_2^t$. In other words, we map sub-constraints with field inputs to sub-constraints with binary bits as input.

Note that we can represent each $\bar\Pi_\ell$ as a Boolean circuit of size $2^{O(t)}$ in time $2^{O(t)}$.
\item For each coordinate $j\in [d]$, we define $\kappa(j)=\cbra{\ell\in[m]\colon j\in Q_{\ell}}$ as the set of sub-constraints applied on the $j$-th coordinate.
\item For each $S\subseteq[m]$, we build a Boolean circuit $\Ccal_S$ to compute the conjunction of the sub-constraints $\bar\Pi_\ell$ for $\ell\in S$.

Formally, $\Ccal_S$ is the Boolean function mapping $\Fbb_2^{k\cdot t}$ to $\bin$ such that 
$$
\Ccal_S(y_1,\ldots,y_k)=\bigwedge_{\ell\in S}\bar\Pi_\ell(y_{2\ell-1},y_{2\ell})=\bigwedge_{\ell\in S}\Pi_\ell\pbra{\chi^{-1}(y_{2\ell-1}),\chi^{-1}(y_{2\ell})},
$$
where each $y_i\in\Fbb_2^t$ is the binary representation of a coordinate of the original $x_i$ via additive isomorphism $\chi$.

By adding dummy gates, we assume each $\Ccal_S$ has exactly $c=k\cdot2^{O(t)}=k\cdot|\Fbb|^{O(1)}$ gates, since the circuit representation of each $\bar\Pi_\ell$ has size $2^{O(t)}$.
In addition, by rearranging indices, we assume the first $k\cdot t$ gates are input gates corresponding to $(y_1,\ldots,y_k)$.
The construction of each $\Ccal_S$ can also be done in time $k\cdot|\Fbb|^{O(1)}$.
\end{itemize}
We remark that the reason why we convert $\Fbb$ into binary bits is that \textsc{Quadeq} can only handle binary circuits and sticking with $\Fbb$ will require equation systems of higher degree to preserve the satisfiability, which complicates the analysis.

Now the satisfiability of the \vcsp{} instance $G$ is equivalent to the satisfiability of the Boolean circuits $\Ccal_S$'s. This is formalized in \Cref{clm:parallel_circuit}.
For convenience, for each assignment $\sigma\colon V\to\Fbb^d$ of $G$ and each coordinate $j\in[d]$, we define $\sigma^j\colon V\to\Fbb$ as the sub-assignment of $\sigma$ on the $j$-th coordinate of all variables in $V$. Note that $\sigma$ and $\sigma^j$ can be equivalently viewed as vectors in $(\Fbb^d)^k$ and $\Fbb^k$ respectively.

\begin{claim}\label{clm:parallel_circuit}
Let $\sigma\colon V\to\Fbb^d$ be an assignment of $V$. 
Then $\sigma$ is a solution of $G$ iff $\Ccal_{\kappa(j)}(y^j_1,\ldots,y^j_k)=1$ holds for every $j\in[d]$, where each $y^j_i=\chi(\sigma^j(x_i))$ is the binary representation of $\sigma^j(x_i)$.
\end{claim}

At this point, we appeal to the \textsc{Quadeq} problem to further encode the satisfiability of each $\Ccal_S$ as the satisfiability of a quadratic equation system.
For each $S\subseteq[m]$, we construct a \textsc{Quadeq} instance $\Gamma_S$ by \Cref{fct:circuit_to_quadeq}, which consists of matrices $D_{S,1},\ldots,D_{S,q}\in\Fbb_2^{c\times c}$ and bits $b_{S,1},\ldots,b_{S,q}\in\Fbb_2$ with $q=O(c)$.
Note that we assume each $\Gamma_S$ shares the same quantity $q$ by padding dummy quadratic equations like $D\equiv0^{c\times c},b\equiv0$.
In addition, by the ``moreover'' part of \Cref{fct:circuit_to_quadeq}, we assume the first $k\cdot t$ bits in an assignment $u\in\Fbb_2^c$ correspond to the input gates of the circuit $\Ccal_S$; and the rest corresponds to values of other gates in $\Ccal_S$.

Then \Cref{clm:parallel_quadeq} establishes the conversion from $G$ to a parallel {\sc Quadeq} instance.

\begin{claim}\label{clm:parallel_quadeq}
Let $\sigma$ be an assignment of $V$. Recall that $\sigma^j$ is the sub-assignment of $\sigma$ on the $j$-th coordinate.
Let $(D_{S,1},\ldots,D_{S,q},b_{S,1},\ldots,b_{S,q})$ be the {\sc Quadeq} instance $\Gamma_S$ for $\Ccal_S$.

Then $\sigma$ is a solution of $G$ iff $(u^j)^\top D_{\kappa(j),i}u^j=b_{\kappa(j),i}$ holds for all $j\in[d]$ and $i\in[q]$, where each $u^j\in\Fbb_2^c$ is some vector with the first $k\cdot t$ bits equal to $\sigma^j$.
In addition, we have $c=k\cdot|\Fbb|^{O(1)}$ and $q=k\cdot|\Fbb|^{O(1)}$ here.

Moreover, $u^j$ represents the values of gates in $\Ccal_{\kappa(j)}$ given as input the first $k\cdot t$ bits of $u^j$.
\end{claim}

We remark that the computation so far is very efficient and runs in FPT time since $|\Fbb|\le h(k)$.

\subsection{Designing Parallel PCPPs for Parallel {\sc Quadeq}}\label{sec:pcpp_for_quadeq}

In light of \Cref{clm:parallel_quadeq}, we now aim to generalize the PCP verifier of \textsc{Quadeq} to the parallel setting to verify the computation on $d$ coordinates simultaneously.
The key observation is that, there are only $2^m=2^{O(k)}$ many different \textsc{Quadeq} instances in \Cref{clm:parallel_quadeq} since $m\le k/2$.
Thus, by tensoring up the proofs for different instances, we can access different positions in different proofs at the same time while still in FPT time.

Recall that for every $j \in [d]$, the set $\kappa(j)\subseteq[m]$ is the set of sub-constraints applied on the $j$-th coordinate. 
We abuse the notation to view each $S\subseteq[m]$ as an integer in $[2^m]$ by some natural bijection.
For each $S\in[2^m]$, we recall that $(D_{S,1},\ldots,D_{S,q},b_{S,1},\ldots,b_{S,q})$ is the $\textsc{Quadeq}$ instance $\Gamma_S$ (see \Cref{clm:parallel_quadeq}) reduced from circuit $\mathcal C_S$ (see \Cref{clm:parallel_circuit}). 
We also recall that $\sigma^j(x_i) \in \mathbb F$ is the $j$-th entry of $\sigma(x_i) \in \mathbb F^d$.
For clarity, we use $\pwh_2$ to denote the parallel Walsh-Hadamard encoding with field $\Fbb_2$ and reserve $\pwh$ for the parallel Walsh-Hadamard encoding with field $\Fbb$.

The verifier $A$ is defined as follows.

\paragraph*{Input of $A$.} 
The verifier $A$ takes as input $\pi_1 \circ \pi_2$, where:
\begin{itemize}
    \item $\pi_1$ has length $|\mathbb F|^k$ and alphabet $\mathbb F^d$.
    
    It is supposed to be $\pwh(\sigma)$ for an assignment $\sigma$ to the variables of $G$.
    
\begin{figure}[ht]
    \centering
    \includegraphics[width=\textwidth]{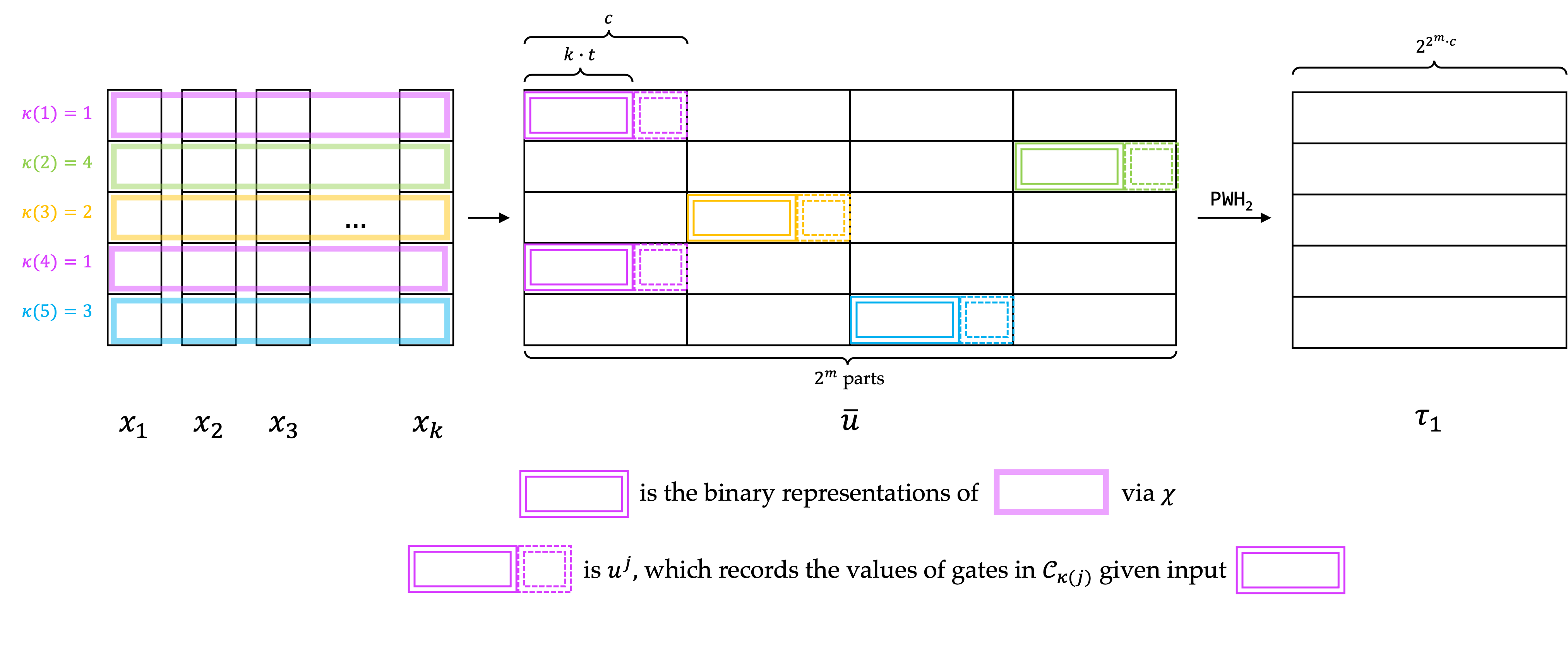}
    \caption{The construction of $\overline u$ and its encoding $\tau_1$. Here $d=5$ and $m=2$.}
    \label{fig:pcpp-parallel-1}
\end{figure}

    \item $\pi_2$ consists of two parts: a $2^{2^m\cdot c}$-length string $\tau_1$ with alphabet $\mathbb F_2^d$ and a $2^{2^m \cdot c^2}$-length string $\tau_2$ with alphabet $\mathbb F_2^d$.
    
    $\tau_1$ and $\tau_2$ are supposed to be $\pwh_2(\overline u)$ and $\pwh_2(\overline w)$ for some $\overline u \in (\mathbb F_2^d)^{2^m \cdot c}$ and $\overline w \in (\mathbb F_2^d)^{2^m \cdot c^2}$ constructed as follows: for each $j \in [d]$, we use $u^j \in \mathbb F_2^c,w^j \in \mathbb F_2^{c \times c}$ to denote the proof\footnote{Technically this proof is for {\sc Quadeq} instance $\Gamma_{\kappa(j)}$. But due to the correspondence in \Cref{fct:circuit_to_quadeq} (or \Cref{clm:parallel_quadeq}), we view it as a proof for the satisfiability of $\Ccal_{\kappa(j)}$. In fact, $u^j$ is the values of gates in $\Ccal_{\kappa(j)}$ and $w^j = u^j(u^j)^\top$.} that the binary representations of $\sigma^j$ satisfy the circuit $\mathcal C_{\kappa(j)}$. 
    For $\overline u$ (resp., $\overline w$), we place $u^j$ (resp., $w^j$) on the $j$-th coordinate and at the $\kappa(j)$-th length-$c$ (resp., length-$c^2$) part, and leave all remaining parts zero. 
    See \Cref{fig:pcpp-parallel-1} for an illustration.
\end{itemize}

We remark that the alphabet of the verifier $A$ here has different alphabets ($\Fbb^d$ and $\Fbb_2^d$) for $\pi_1$ and $\pi_2$.
This is convenient for stating the tests and the analysis. 
To make it consistent with the definition of PPCPP (\Cref{def:pcpp}), we can simply perform a black-box reduction that equips $\pi_2$ with alphabet $\Fbb^d$ as well but rejects if any query result during the test is not from $\bin^d\cong\Fbb_2^d$.

To get a sense of the detail inside $\tau_1$ (or $\tau_2$), we refer to \Cref{fig:pcpp-parallel-2} where $\tau_1$ supposedly faithfully stores $\pwh_2(\bar u)$.
This will be helpful in understanding the tests of $A$ below.

\begin{figure}[ht]
    \centering
    \includegraphics[width=\textwidth]{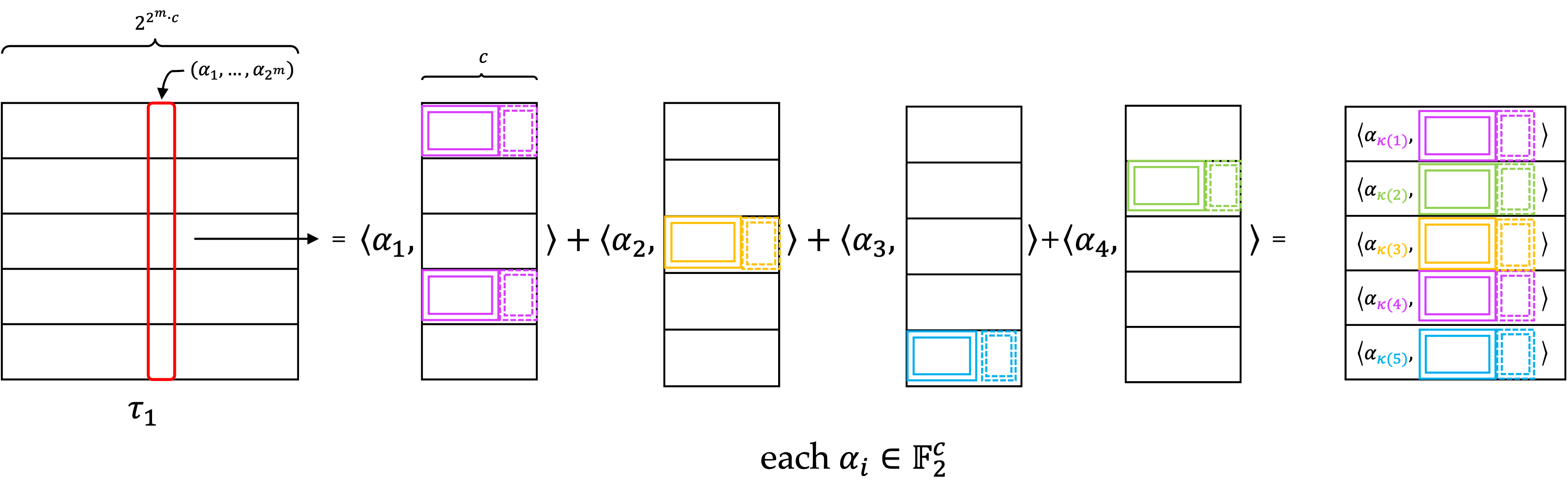}
    \caption{An illustration of $\tau_1[\alpha_1,\ldots,\alpha_{2^m}]$.}
    \label{fig:pcpp-parallel-2}
\end{figure}

\paragraph*{Verification Procedure of $A$.}
The verifier $A$ selects one of the following eight tests with equal probability. For ease of understanding, we group the tests according to their functions.

\newcounter{parallel_test}
\addtocounter{parallel_test}{1}

\begin{itemize}
    \item \textsc{Linearity Test.}
    \begin{enumerate}[label=\textbf{(P\arabic{parallel_test})}]
        \item \label{itm:pcpp-parallel-1} Pick uniformly random $\alpha,\beta \in \mathbb F^k$ and check if $\pi_1[\alpha]+\pi_1[\beta]=\pi_1[\alpha+\beta]$ with three queries. 
        \addtocounter{parallel_test}{1}

        \item \label{itm:pcpp-parallel-2} Pick uniformly random $\alpha,\beta \in \mathbb F_2^{2^m\cdot c}$ and check if $\tau_1[\alpha]+\tau_1[\beta]=\tau_1[\alpha+\beta]$ with three queries. 
        \addtocounter{parallel_test}{1}

        \item \label{itm:pcpp-parallel-3} Pick uniformly random $\alpha,\beta \in \mathbb F_2^{2^m\cdot c^2}$ and check if $\tau_2[\alpha]+\tau_2[\beta]=\tau_2[\alpha+\beta]$ with three queries. 
        \addtocounter{parallel_test}{1}
    \end{enumerate}
    These three tests ensure that $\pi_1,\tau_1,\tau_2$ are close to $\pwh(\sigma),\pwh_2(\overline u),\pwh_2(\overline w)$ for some $\sigma \in (\mathbb F^d)^{k}, \overline u \in (\mathbb F_2^d)^{2^m \cdot c}$ and $\overline w \in (\mathbb F_2^d)^{2^m\cdot c^2}$, respectively.
    \item \textsc{Zero Test.}
    \begin{enumerate}[label=\textbf{(P\arabic{parallel_test})}]
        \item \label{itm:pcpp-parallel-4} Take a random subset $T$ of $[2^m]$, generate a random $\alpha_i\in \mathbb F_2^c$ for each $i \in T$ and set $\alpha_i=0$ for each $i \notin T$. 
        Then pick uniformly random $\beta_1,\ldots,\beta_{2^m} \in \mathbb F_2^c$ and obtain $v:=\tau_1[\beta_1,\ldots,\beta_{2^m}]+\tau_1[\alpha_1+\beta_1,\ldots,\alpha_{2^m}+\beta_{2^m}]$ by two queries. Reject if for some $j \in [d]$, we have $\kappa(j) \notin T$ but the $j$-th coordinate of $v$ is non-zero. 

        \addtocounter{parallel_test}{1}

        \item \label{itm:pcpp-parallel-5} 
        Take a random subset $T$ of $[2^m]$, generate a random $\alpha_i\in \mathbb F_2^{c\times c}$ for each $i \in T$ and set $\alpha_i=0$ for each $i \notin T$. 
        Then pick uniformly random $\beta_1,\ldots,\beta_{2^m} \in \mathbb F_2^{c\times c}$ and obtain $v:=\tau_2[\beta_1,\ldots,\beta_{2^m}]+\tau_2[\alpha_1+\beta_1,\ldots,\alpha_{2^m}+\beta_{2^m}]$ by two queries. Reject if for some $j \in [d]$, we have $\kappa(j) \notin T$ but the $j$-th coordinate of $v$ is non-zero. 
        \addtocounter{parallel_test}{1}
    \end{enumerate}
    These two tests ensure that $\overline u$ and $\overline w$ are of the forms we want, i.e., for every $S\in [2^m]$, the $S$-th length-$c$ part (respectively, length-$c^2$ part) has non-zero values on the $j$-th coordinate only if $\kappa(j)=S$.
    
    \item \textsc{Tensor Test.}
    \begin{enumerate}[label=\textbf{(P\arabic{parallel_test})}]
        \item \label{itm:pcpp-parallel-6} Pick uniformly random $r_1,\ldots,r_{2^m},r'_1,\ldots,r'_{2^m} \in \mathbb F_2^c$ and $y_1,\ldots,y_{2^m} \in \mathbb F_2^{c\times c}$, and check whether 
        \begin{equation}
        \label{eq:pcpp-parallel-1}
            \tau_1[r_1,\ldots,r_{2^m}] \odot \tau_1[r'_1,\ldots,r'_{2^m}] = \tau_2[y_1,\ldots,y_{2^m}] + \tau_2[y_1+r_1r_1'^\top,\ldots,y_{2^m}+r_{2^m}r_{2^m}'^\top]
        \end{equation}
        with four queries, where $\odot$ is the coordinate-wise multiplication. 
        \addtocounter{parallel_test}{1}
    \end{enumerate}
    This performs the {\sc Tensor Test} (\Cref{itm:quadeq_tensor}) of \textsc{Quadeq} on all $d$ coordinates simultaneously.
    \item \textsc{Constraint Test.}
    \begin{enumerate}[label=\textbf{(P\arabic{parallel_test})}]
        \item \label{itm:pcpp-parallel-7} Pick a random subset $H$ of $[q]$ and uniformly random $\beta_1,\ldots,\beta_{2^m}\in\Fbb_2^{c\times c}$.
        For each $S\in[2^m]$, define $\alpha_S = \sum_{z \in H} D_{S,z} \in \mathbb F_2^{c\times c}$.
        Obtain $y:=\tau_2[\beta_1,\ldots,\beta_{2^m}]+\tau_2[\alpha_1+\beta_1,\ldots,\alpha_{2^m}+\beta_{2^m}]$ by two queries and reject if for some $j \in [d]$, the $j$-th coordinate does not equal to $\sum_{z \in H} b_{\kappa(j),z}$. 
        \addtocounter{parallel_test}{1}
    \end{enumerate}
    This performs the {\sc Constraint Test} (\Cref{itm:quadeq_constraint}) of \textsc{Quadeq} on all $d$ coordinates simultaneously, where on the $j$-th coordinate we check the constraints with respect to $\mathcal C_{\kappa(j)}$.
    \item \textsc{Consistency Test.}
    \begin{enumerate}[label=\textbf{(P\arabic{parallel_test})}]
        \item \label{itm:pcpp-parallel-8}
        Pick a random subset $D$ of $[k]$ and a uniformly random $\beta \in \mathbb F^k$. 
        Pick a random linear function $\psi\colon\mathbb F_2^{t} \to \mathbb F_2$ and uniformly random $\xi_1,\ldots,\xi_{2^m} \in \mathbb F_2^c$.
        Define $\alpha\in\Fbb^k$ to be the indicator vector of $D$, i.e., $\alpha_i=1$ for $i\in D$ and $\alpha_i=0$ for $i\notin D$.
        Let 
        $$
        \gamma =(\psi(1,0,\ldots,0),\psi(0,1,\ldots,0),\ldots,\psi(0,0,\ldots,1)) \in \mathbb F_2^{t}
        $$
        and
        $$
        \eta  =(\underbrace{\gamma_1, \ldots,\gamma_{k}}_{k\text{ of $t$ bits}},\underbrace{0,\ldots,0}_{\text{remaining }c-kt\text{ bits}}) \in \mathbb F_2^c
        \quad\text{where}\quad
        \gamma_i=
        \begin{cases}
        \gamma & \text{if } i \in D, \\
        0^{t} & \text{otherwise.}
        \end{cases}
        $$
        Then check if 
        \begin{equation}\label{eq:pcpp-parallel-2}
        \psi\circ\chi(\pi_1[\beta]+\pi_1[\alpha+\beta])=\tau_1[\xi_1,\ldots,\xi_{2^m}]+\tau_1[\eta+\xi_1,\ldots,\eta+\xi_{2^m}],  
        \end{equation}
        where $\psi\circ\chi\colon\Fbb\to\Fbb_2$ is applied coordinate-wise.
    \end{enumerate}
    This test checks if for every $j \in [d]$, the first $k \cdot t$ bits in $u^j$ equal to the binary representations of $\sigma^j$ specified by the isomorphism $\chi$.
\end{itemize}

\subsection{Analysis of Parallel PCPPs}

In this subsection, we prove \Cref{prop:pcpp-parallel-constraint} with the following three lemmas (\Cref{lem:pcpp-parallel-1}, \Cref{lem:pcpp-parallel-2}, and \Cref{lem:pcpp-parallel-3}), which are devoted to bound the parameters, and show completeness and soundness, respectively.

\begin{lemma}[Parameters]\label{lem:pcpp-parallel-1}
    The verifier $A$ takes as input two proofs $\pi_1$ and $\pi_2$, where $\pi_1$ has length $|\Fbb|^k$ and $\pi_2$ has length at most $f(k)=2^{2^k\cdot|\Fbb|^{O(1)}}$. $A$ then uses at most $g(k)=2^{2^k\cdot|\Fbb|^{O(1)}}$ randomness, and queries at most four positions of the proofs. Furthermore, the list of queries made by $A$ can be generated in FPT time.   
\end{lemma}
\begin{proof}
    The length of $\pi_1$ is $|\Fbb|^k$ by definition, and the length of $\pi_2$ is $2^{2^m\cdot c}+2^{2^m\cdot c^2}\le2^{2^k\cdot|\Fbb|^{O(1)}}$, where we recall that $m\le k/2$ and $c=k\cdot|\Fbb|^{O(1)}$.

    The amount of randomness is calculated as follows. \Cref{itm:pcpp-parallel-1} has $|\Fbb|^{2k}$ uniform possibilities, \Cref{itm:pcpp-parallel-2} has $2^{2c\cdot2^m}$, \Cref{itm:pcpp-parallel-3} has $2^{2c^2\cdot2^m}$, \Cref{itm:pcpp-parallel-4} has\footnote{The second $2^{c\cdot2^m}$ comes from additionally sampling random elements for $i\notin T$ for padding to make sure they are uniform possibilities. Similar for \Cref{itm:pcpp-parallel-5}.} $2^{2^m}\cdot2^{c\cdot2^m}\cdot2^{c\cdot2^m}$, \Cref{itm:pcpp-parallel-5} has $2^{2^m}\cdot2^{c^2\cdot2^m}\cdot2^{c^2\cdot2^m}$, \Cref{itm:pcpp-parallel-6} has $2^{2c\cdot2^m}\cdot2^{c^2\cdot2^m}$, \Cref{itm:pcpp-parallel-7} has $2^q\cdot2^{c^2\cdot2^m}$, and \Cref{itm:pcpp-parallel-8} has\footnote{The second $2^t$ comes from the randomness in $\psi$, which is a random \emph{linear} function from $\Fbb_2^t$ to $\Fbb_2$.} $2^k\cdot|\Fbb|^k\cdot2^t\cdot2^{c\cdot2^m}$.
Recall that $2^t=|\Fbb|$, $m\le k/2$, $c=k\cdot|\Fbb|^{O(1)}$, and $q=k\cdot|\Fbb|^{O(1)}$.
Hence we may duplicate integer multiples for each of them and assume that they all have $2^{2^k\cdot|\Fbb|^{O(1)}}$ uniform possibilities.
Then the total randomness sums up to $8\cdot2^{2^k\cdot|\Fbb|^{O(1)}}\le g(k)$ as desired.

It's easy to see that $A$ makes at most four queries in any case, and the list of queries under all randomness can be generated in FPT time.
\end{proof}

\begin{lemma}[Completeness]\label{lem:pcpp-parallel-2}
    Suppose there is a solution $\sigma:V \to \Fbb^d$ of $G$, then there is a proof $\pi_1\circ\tau_1\circ\tau_2$ which $A$ accepts with probability 1.
\end{lemma}
\begin{proof}
By \Cref{clm:parallel_circuit}, for each $j \in [d]$, $\mathcal C_{\kappa(j)}$ outputs $1$ when taking the binary representation of $\sigma^j$, i.e., $(\chi(\sigma^j(x_1)),\ldots,\chi(\sigma^j(x_k)))$, as input. 
Thus by \Cref{clm:parallel_quadeq}, for each $j \in [d]$, we have a solution $u^j$ to the {\sc Quadeq} instance $\Gamma_{\kappa(j)}$, where the first $k\cdot t$ bits of $u^j$ equal to $\chi(\sigma^j(x_1)),\ldots,\chi(\sigma^j(x_k))$. 

We set $\pi_1=\pwh(\sigma(x_1),\ldots,\sigma(x_k))$, $\tau_1=\pwh_2(\overline u)$, and $\tau_2=\pwh_2(\overline w)$ where $\overline u \in \left(\mathbb F_2^d\right)^{2^m\cdot c}$ and $\overline w \in \left(\mathbb F_2^d\right)^{2^m \cdot c^2}$ defined to be consistent with \Cref{fig:pcpp-parallel-1}:
\begin{itemize}
    \item For every $j \in [d]$, the $j$-th coordinate of $\overline u$, viewed as a length-$(2^m\cdot c)$ binary string, has $u^j$ on the $\kappa(j)$-th length-$c$ part, and zero everywhere else.
    \item For every $j \in [d]$, the $j$-th coordinate of $\overline w$, viewed as a length-$(2^m \cdot c^2)$ binary string, has $u^j(u^j)^\top$ on the $\kappa(j)$-th length-$c^2$ part, and zero everywhere else.
\end{itemize}

Since $\pi_1,\tau_1$ and $\tau_2$ are all parallel Walsh-Hadamard codewords, they pass the linearity tests in \Cref{itm:pcpp-parallel-1}, \Cref{itm:pcpp-parallel-2}, \Cref{itm:pcpp-parallel-3} naturally.

Given $\overline u$ defined as above, any query to $\tau_1[\alpha_1,\ldots,\alpha_{2^m}]$, where each $\alpha_i \in \mathbb F_2^c$, gives us a vector $v \in \mathbb F_2^d$, whose $j$-th coordinate stores $\langle \alpha_{\kappa(j)},u^{j}\rangle$. See \Cref{fig:pcpp-parallel-2} for an illustration. 
Thus for any subset $T$ of $[2^m]$, if we set $\alpha_i=0$ for all $i \notin T$ as in \Cref{itm:pcpp-parallel-4} and \Cref{itm:pcpp-parallel-5}, the resulting $v$ equals zero on all coordinates $j\in[d]$ with $\kappa(j) \notin T$, which passes the tests.

To verify \Cref{itm:pcpp-parallel-6} and \Cref{itm:pcpp-parallel-7}, we simply observe that
\begin{itemize}
    \item For every $j \in [d]$ and $r,r' \in \mathbb F_2^c$, we have $\langle w^j,rr'^\top\rangle=\langle u^j(u^j)^\top,rr'^\top\rangle = (r^\top u^j) (r'^\top u^j)$.
    \item For every $j \in [d]$ and $H \subseteq [q]$, we have
    \[\left\langle w^j,\sum_{z \in H} D_{\kappa(j),z} \right\rangle= (u^j)^\top\left( \sum_{z \in H} D_{\kappa(j),z}\right) u^j=\sum_{z \in H} (u^j)^\top D_{\kappa(j),z} u^j = \sum_{z \in H}b_{\kappa(j),z},\] 
    since $u^j$ is a solution to the \textsc{Quadeq} instance $\Gamma_{\kappa(j)}$.
\end{itemize}

For \Cref{itm:pcpp-parallel-8}, on the left hand side of \Cref{eq:pcpp-parallel-2}, we have 
\[
\psi\circ\chi(\pi_1[\beta]+\pi_1[\alpha+\beta])=\psi\circ\chi(\pi_1[\alpha])=\psi\circ\chi\left(\sum_{i \in S} \sigma(x_i)\right) = \sum_{i \in S}\psi(\chi(\sigma(x_i))),
\]
where $\psi\circ\chi$ is applied coordinate-wise and the second equality is due to the linearity of $\psi$ and the fact that $\chi$ is a additive isomorphism. 
On the right hand side of \Cref{eq:pcpp-parallel-2}, by our choice of $\eta$ and the fact that the first $k\cdot t$ bits of each $u^j$ are just $(\chi(\sigma^j(x_1)),\ldots,\chi(\sigma^j(x_k)))$,  we also get $\sum_{i \in S} \psi(\chi(\sigma(x_i)))$.     
\end{proof}

\begin{lemma}[Soundness]\label{lem:pcpp-parallel-3}
    Suppose there is a proof $\pi_1\circ \tau_1\circ \tau_2$ which $A$ accepts with probability at least $1-\varepsilon$, then there is a solution $\sigma$ to $G$ such that $\Delta(\pi_1,\pwh(\sigma)) \le 48 \varepsilon$.
\end{lemma}
\begin{proof}
Given such a proof, each individual test passes with probability at least $1-8\varepsilon$. By the soundness of BLR testing (\Cref{thm:pwh_test}), passing the linearity test in \Cref{itm:pcpp-parallel-1} with probability at least $1-8\varepsilon$ implies there exists $\sigma \in (\mathbb F^d)^k$ such that $\Delta(\pi_1,\pwh(\sigma))\le 48\varepsilon$. 

Similarly for \Cref{itm:pcpp-parallel-2} and \Cref{itm:pcpp-parallel-3}, $\tau_1,\tau_2$ are $(48\varepsilon)$-close to $\pwh_2(\overline u)$ and $\pwh_2(\overline w)$ for some $\overline u \in (\mathbb F_2^d)^{2^m\cdot c}$ and $\overline w \in (\mathbb F_2^d)^{2^m \cdot c^2}$ respectively. 

Next we prove that, for every $j \in [d]$, the $j$-th coordinate of $\overline u$, viewed as a length-$(2^m\cdot c)$ binary string, has non-zero values only in the $\kappa(j)$-th length-$c$ part. Suppose it is not, and it is non-zero on the $\ell$-th length-$c$ part for some $\ell \neq \kappa(j)$. 
Then in \Cref{itm:pcpp-parallel-4}, with probability $\frac{1}{4}$, we have $\ell \in T$ but $\kappa(j) \notin T$. 
Now for any $\{\alpha_{i}\}_{i \ne \ell}$, by the random subsum principle (\Cref{lem:rsp}), for at least $\frac{1}{2}$ of the choices of $\alpha_{\ell}$, the $j$-th coordinate of $\pwh_2(\overline u)[\alpha_1,\ldots,\alpha_{2^m}]$ is non-zero.
Furthermore, since $\kappa(j) \notin T$, our test will reject whenever the $j$-th entry is non-zero. Thus, as long as the local correction procedure correctly decodes $\pwh_2(\bar u)[\alpha_1,\ldots,\alpha_{2^m}]$, which happens with probability at least $1-96\varepsilon$, the test fails. 
The overall rejection probability in this test is then at least $\frac{1}{8}-96 \varepsilon$, which is greater than the $8\varepsilon$ as $\varepsilon< \frac{1}{832}$. Thus a contradiction. Similar analysis also works for \Cref{itm:pcpp-parallel-5}.

At this point, for each $j \in [d]$, define $u^j \in \mathbb F_2^c$ (resp., $w^j \in \mathbb F_2^{c \times c}$) to be the $j$-th coordinate of the $\kappa(j)$-th length-$c$ (resp., length-$c^2$) part of $\overline u$ (resp., $\overline w$). 
By the analysis above, any query $\tau_1[\alpha_1,\ldots,\alpha_{2^m}]$ gives us a vector $v \in \mathbb F_2^d$, whose $j$-th coordinate stores $\langle u^j,\alpha_{\kappa(j)}\rangle$; and the same holds for $\tau_2$. See \Cref{fig:pcpp-parallel-2} for an illustration.
We then prove that $w^j = u^j(u^j)^\top$ holds for every $j\in[d]$, and ${u^j}^\top D_{\kappa(j),z} u^j=b_{\kappa(j),z}$ holds for every $j \in [d],z \in [q]$. 
This shows that they form a solution of the quadratic equation system in \Cref{clm:parallel_quadeq}.
\begin{itemize}
\item 
If for some $j \in [d]$, we have $w^j \neq u^j \otimes u^j$.
Then by \Cref{lem:rsp} twice, for at least $\frac{1}{4}$ fraction of choices of $(r,r')$, we have $r^\top w^j r'\neq (r^\top u^j)(r'^\top u^j)$. 
Suppose the four queried points in \Cref{itm:pcpp-parallel-6} are indeed as if on $\pwh_2(\bar u),\pwh_2(\bar w)$, which happens with probability at least $1-192\varepsilon$. Then the left hand side of \Cref{eq:pcpp-parallel-1} is a vector $v$ whose $j$-th coordinate is $(r^\top u^j)(r'^\top u^j)$, while the right hand side of \Cref{eq:pcpp-parallel-1} is a vector $v'$ whose $j$-th coordinate is $r^\top w^j r'$. 
Thus $v \neq v'$ and \Cref{itm:pcpp-parallel-6} rejects.
Now that the rejection probability is at least $\frac{1}{4}-192\varepsilon$, it is greater than $8\varepsilon$ when $\varepsilon<\frac{1}{800}$, which provides a contradiction. 
\item
If for some $j \in [d], z \in [q]$, we have ${u^j}^\top D_{\kappa(j),z} u^j \neq b_{\kappa(j),z}$.
By \Cref{lem:rsp}, for $\frac{1}{2}$ fraction of choices of $H \subseteq [q]$, we have $\sum_{z \in H} {u^j}^\top D_{\kappa(j),z}u^j \neq \sum_{z \in H} b_{\kappa(j),z}$. Suppose the two queried points in \Cref{itm:pcpp-parallel-7} are indeed as if on $\pwh_2(\bar w)$, which happens with probability at least $1-96\varepsilon$. Then \Cref{itm:pcpp-parallel-7} rejects.
Now that the rejection probability is at least $\frac{1}{2}-96 \varepsilon$, it is greater than $8\varepsilon$ as $\varepsilon<\frac{1}{208}$, which provides a contradiction.
\end{itemize}

Finally, we prove that for every $j \in [d]$, the first $k\cdot t$ bits in $u^j$, which we denote as $(u^j_1, \ldots, u^j_{k}) \in \mathbb F_2^{k \cdot t}$, equal to $(\chi(\sigma^j(x_1)),\ldots,\chi(\sigma^j(x_k)))$. This shows that the binary representation of $\sigma$ certifies the satisfiability of the circuits $\Ccal_S$'s as well as the {\sc Quadeq} instances $\Gamma_S$'s by \Cref{clm:parallel_circuit} and \Cref{clm:parallel_quadeq}, which means $\sigma$ is a solution of $G$.

Suppose for some $i \in [k],j \in [d]$, we have $\chi(\sigma^j(x_i)) \neq u^j_i$. 
Since $\psi$ is a random linear function mapping to $\Fbb_2$, with probability $\frac12$, they still differ after $\psi$ applied on. 
Then by \Cref{lem:rsp}, for $\frac{1}{2}$ of the choices of $S\subseteq [k]$, we have
\[
\sum_{i \in S} \psi(\chi(\sigma^j(x_i))) \neq \sum_{i \in S} \psi(u^j_i).
\]
Suppose the four queried points in \Cref{itm:pcpp-parallel-8} are indeed as if on $\pwh(\sigma),\pwh_2(\bar u)$, which happens with probability at least $1-192\varepsilon$. Then the left hand side of \Cref{eq:pcpp-parallel-2} is
\[
\psi\circ\chi\left(\sum_{i \in S} \sigma(x_i)\right) = \sum_{i \in S}\psi(\chi(\sigma(x_i))),
\]
where $\psi\circ\chi$ is applied coordinate-wise and the equality holds due to the linearity of $\psi$ and the fact that $\chi$ is a additive isomorphism. 
The right hand side of \Cref{eq:pcpp-parallel-2} is $\sum_{i \in S} \psi(u_i^j)$ by our construction of $\eta$. Therefore, \Cref{itm:pcpp-parallel-8} rejects with probability at least $\frac{1}{4}-192\varepsilon$, which is greater than $8\varepsilon$ when $\varepsilon<\frac{1}{800}$, leading to a contradiction again.

In conclusion, we have shown that if $A$ accepts with probability at least $1-\varepsilon$, then $\pi_1$ must be $(48\varepsilon)$-close to $\pwh(\sigma)$, where $\sigma$ is a solution of $G$.
\end{proof}

\Cref{prop:pcpp-parallel-constraint} follows from a combination of \Cref{lem:pcpp-parallel-1}, \Cref{lem:pcpp-parallel-2} and \Cref{lem:pcpp-parallel-3}.
\section{Parallel PCPPs for Vector-Valued CSPs with Linear Constraints}\label{sec:linear-pcpp}
This section is devoted to proving~\Cref{prop:pcpp-linear-constraint}, which we recall below.

\begin{proposition*}[\Cref{prop:pcpp-linear-constraint}~Restated]
Let $h$ and $m$ be two computable functions.
Let $G$ be a \vcsp{} instance with $k$ variables where (1) the alphabet is $\Fbb^d$ and $|\Fbb|\le h(k)$, (2) all constraints are linear constraints, and (3) there are at most $m(k)$ constraints.
Then for every $\varepsilon\in \left(0, \frac1{400}\right)$, there is a $(4,24\varepsilon,\varepsilon,f(k)=|\Fbb|^{k\cdot m(k)}, g(k)=|\Fbb|^{8k\cdot m(k)})$-PPCPP verifier for $G$.
\end{proposition*}

\subsection{Construction of Parallel PCPPs}

Fix a \vcsp{} instance $G = (V,E,\Sigma,\{\Pi_e\}_{e\in E})$ from \Cref{prop:pcpp-linear-constraint}. 
Recall that $k = |V|$ and we set $m = |E|\le m(k)$.
By \Cref{def:vcsp}, since all constraints are linear, for each constraint $e\in E$ we denote
\begin{itemize}
    \item its two endpoints by $u_e$ and $v_e$,
    \item the matrix for this linear constraint by $M_e\in \Fbb^{d\times d}$,
    \item the semantics of this constraint by $\Pi_e(u_e,v_e) = \indicator_{u_e = M_ev_e}$. 
\end{itemize}

For ease of presentation, we call $u_e$ the \textit{head} of the constraint $e$, and $v_e$ the \textit{tail} of $e$, respectively.

Our construction of the PPCPP verifier $A$ is similar to the Walsh-Hadamard-based one in \cite{BGH05}, with an additional introduction of some subtle auxilary variables.

\paragraph*{Auxilary Variables.} Label variables $V$ by $\{1,2,\dots,k\}$ and constraints by $\{1,2,\dots,m\}$. 
For every $p\in V$ and $e\in E$, we define an auxiliary variable $z_{p,e}$ with alphabet $\Fbb^d$. Given an assignment $\sigma(p)$ to the variable $p$, the assignment to $z_{p,e}$ should equal $z_{p,e} = M_e\sigma(p)$ \footnote{Here we abuse the notation and use $z_{p,e}$ also to denote the value assigned to it.}. 

Note that we introduce an auxilary variable for every pair $(p,e) \in V \times E$, \emph{even if $e$ is not adjacent to $p$.} This way, we can check both the inner constraints $z_{p,e} = M_e\sigma(p)$ and the conjunction of all linear constraints $\sigma(u_e)=z_{v_e,e}$ with constant queries, soundness, and proximity. 

Below, we describe the details of the PPCPP verifier $A$ for $G$.

\paragraph*{Input of $A$.} 
The verifier $A$ takes as input $\pi_1\circ \pi_2$, where:
\begin{itemize}
\item $\pi_1$ is indexed by vectors in $\Fbb^k$ and has alphabet $\Fbb^d$. 
It is supposed to be $\pwh(\sigma)$, the parallel Walsh-Hadamard encoding of an assignment $\sigma$ to $V$.
\item $\pi_2$ is indexed by vectors in $\Fbb^{km}$ and has alphabet $\Fbb^d$. 
It is supposed to be the parallel Walsh-Hadamard encoding of the collection $\{z_{p,e}\}_{p\in V,e\in E}$, treated as a vector of $(\Fbb^{d})^{km}$.
\end{itemize}

\newcounter{linear_test}
\addtocounter{linear_test}{1}

\paragraph*{Verification Procedure of $A$.} 
Here is how $A$ verifies whether $\pi_1$ is close to $\pwh(\sigma)$ for some solution $\sigma$ of $G$. With equal probability, $A$ selects one of the following four tests:
\begin{itemize}
    \item \textsc{Linearity Test.}
    \begin{enumerate}[label=\textbf{(L\arabic{linear_test})}]
    \item\label{itm:pcpp-linear-1} Pick uniformly random $a_1,a_2\in \Fbb^k$ and check $\pi_1[a_1]+\pi_1[a_2] = \pi_1[a_1+a_2]$ by three queries.  \addtocounter{linear_test}{1}
    \item \label{itm:pcpp-linear-2} Pick uniformly random $b_1,b_2\in \Fbb^{km}$ and check $\pi_2[b_1]+\pi_2[b_2] = \pi_2[b_1+b_2]$ by three queries. \addtocounter{linear_test}{1}    
    \end{enumerate}

        Intuitively, \Cref{itm:pcpp-linear-1} and \Cref{itm:pcpp-linear-2} ensure that both $\pi_1$ and $\pi_2$ are close to a codeword of $\pwh$.

    \item \textsc{Matrix Test.}
    \begin{enumerate}[label=\textbf{(L\arabic{linear_test})}]
    \item\label{itm:pcpp-linear-3}  
    Pick uniformly random $\lambda\in \Fbb^k$ and $\mu\in \Fbb^{m}$ and set $\gamma = (\lambda_1\mu_1,\lambda_1\mu_2,\dots,\lambda_k\mu_{m})\in \Fbb^{km}$. 
    Assume $\mu$ is indexed by constraints $e\in E$ and define matrix $M_0 = \sum_{e\in E}\mu_eM_e$. Note that we can compute $M_0$ efficiently without any query.
    
    Then pick uniformly random $a\in \Fbb^k,b\in \Fbb^{km}$, query $\pi_1[a],\pi_1[a+\lambda],\pi_2[b], \pi_2[b+\gamma]$, and check if 
    \begin{equation}\label{eq:test3}
    \pi_2[b+\gamma]-\pi_2[b] = M_0(\pi_1[a+\lambda]-\pi_1[a]).
    \end{equation}

    \addtocounter{linear_test}{1}
    \end{enumerate}
    Intuitively, \Cref{itm:pcpp-linear-3} ensures that $\pi_2$ encodes the collection $\{z_{p,e}\}_{p\in V,e\in E}$ where all inner constraints $z_{p,e} = M_e\sigma(p)$ are satisfied.
    
    \item \textsc{Constraint Test.}
    \begin{enumerate}[label=\textbf{(L\arabic{linear_test})}]
    \item \label{itm:pcpp-linear-4} Pick uniformly random $\mu\in \Fbb^{m}$ and assume $\mu$ is indexed by constraints $e\in E$. 
    Define a vector $\lambda\in \Fbb^{k}$ by setting $\lambda_p = \sum_{e\in E\colon u_e = p}{\mu_e}$ for $p\in V$, where we assume that $\lambda$ is indexed by vertices $p\in V$. In other words, $\lambda_p$ is the sum of $\mu_e$'s for constraint $e \in E$ whose head is $p$.
    
    In addition, define a vector $\gamma\in \Fbb^{km}$, indexed by a vertex-constraint pair $(p,e)\in V\times E$, by
    \begin{equation*}
    \gamma_{p,e}= \begin{cases}
    \mu_e & v_e = p,\\
    0 & \text{otherwise}.
    \end{cases}
    \end{equation*}
    In other words, $\gamma_{p,e}$ stores $\mu_e$ if the tail of the constraint $e$ is vertex $p$.
    
    Note that the two vectors $\mu$ and $\gamma$ can be computed efficiently without any query.
    
    Then pick uniformly random $a\in \Fbb^k,b\in \Fbb^{km}$, query $\pi_1[a],\pi_1[a+\lambda],\pi_2[b], \pi_2[b+\gamma]$, and check if 
    $$
    \pi_2[b+\gamma]-\pi_2[b] = \pi_1[a+\lambda]-\pi_1[a].
    $$
    \end{enumerate}
    Intuitively, \Cref{itm:pcpp-linear-4} ensures $\sigma(u_e)=z_{v_e,e}$ for every constraint $e \in E$.
\end{itemize}

\subsection{Analysis of Parallel PCPPs}
In this subsection, we prove \Cref{prop:pcpp-linear-constraint} with the following three lemmas (\Cref{lem:pcpp-linear-1}, \Cref{lem:pcpp-linear-2} and \Cref{lem:pcpp-linear-3}), which are devoted to bounding the parameters, and establishing the completeness and soundness of the verifier, respectively.

\begin{lemma}[Parameters]\label{lem:pcpp-linear-1}
    The verifier $A$ takes as input two proofs $\pi_1$ and $\pi_2$, where $\pi_1$ has length $|\Fbb|^k$ and $\pi_2$ has length $f(k)=|\Fbb|^{km}$. $A$ then uses at most $g(k)=|\Fbb|^{8km}$ randomness, and queries at most four positions of the proofs. Furthermore, the list of queries made by $A$ can be generated in FPT time.
\end{lemma}
\begin{proof}
    The length of $\pi_1$ and $\pi_2$ are $|\Fbb|^k$ and $|\Fbb|^{km}$ respectively by definition.

    In terms of randomness, \Cref{itm:pcpp-linear-1} has $|\Fbb|^{2k}$ uniform possibilities, \Cref{itm:pcpp-linear-2} has $|\Fbb|^{2km}$, \Cref{itm:pcpp-linear-3} has $|\Fbb|^{2k+m+km}$, and \Cref{itm:pcpp-linear-4} has $|\Fbb|^{k+m+km}$.
    Note that we may duplicate integer multiples for each of them and assume that they all have $|\Fbb|^{4km}$ uniform possibilities.
    Then the total randomness sums up to $4\cdot|\Fbb|^{4km}\le g(k)$ as desired.
    
    It's easy to see that $A$ makes at most four queries in any case, and the list of queries under all randomness can be generated in FPT time.
\end{proof}

\begin{lemma}[Completeness]\label{lem:pcpp-linear-2}
    Suppose there is a solution $\sigma:V \to \Fbb ^d$ of $G$, then there is a proof $\pi_1 \circ \pi_2$ which $A$ accepts with probability 1.
\end{lemma}
\begin{proof}
Fix such a solution $\sigma$. We assign the value $M_e\sigma(p)$ to the auxiliary variable $z_{p,e}$ and treat $\{z_{p,e}\}_{p\in V, e\in E}$ as a $km$-dimensional vector. 
We set $\pi_1$ as $\pwh(\sigma)$ and $\pi_2$ as $\pwh(\{z_{p,e}\}_{p\in V, e\in E})$. Since both $\pi_{1}$ and $\pi_2$ are codewords of $\pwh$, they naturally pass the  tests in \Cref{itm:pcpp-linear-1} and \Cref{itm:pcpp-linear-2}. 

For \Cref{itm:pcpp-linear-3}, note that for every $\lambda\in \Fbb^k,\mu\in \Fbb^{m},a\in \Fbb^k,b\in \Fbb^{km}$, we have 
\begin{align*}
\pi_2[b+\gamma]-\pi_2[b] &=\pi_2[\gamma]= \sum_{p\in V,e\in E} \lambda_p\mu_e z_{p,e}
\tag{by the definition of $\pi_2$ and $\gamma$}\\
&= \sum_{p,e} \lambda_p\mu_e\cdot M_e \sigma(p)= \left(\sum_{e\in E} \mu_eM_e\right) \left(\sum_{p\in V} \lambda_p\sigma(p)\right)
\tag{by the definition of $z_{p,e}$}\\
&= M_0 \pi_1[\lambda]= M_0 (\pi_1[a + \lambda] - \pi_1[a]),
\tag{by the definition of $M_0$ and $\pi_1$}
\end{align*}
which means that the test in \Cref{itm:pcpp-linear-3} passes. 

Finally, we turn to \Cref{itm:pcpp-linear-4}. For every $\mu\in \Fbb^{m},a\in \Fbb^k,b\in \Fbb^{km}$, we have 
\begin{align*}
\pi_2[b+\gamma]-\pi_2[b] &=\pi_2[\gamma]= \sum_{e\in E} \mu_e z_{v_e,e}= \sum_{e\in E} \mu_e\cdot M_e \sigma(v_e)
\tag{by the definition of $\pi_2,\gamma,z_{v_e,e}$}\\
&= \sum_{e\in E} {\mu_e \sigma(u_e)}
\tag{by $M_e\sigma(v_e)=\sigma(u_e)$ as $\sigma$ is a solution}\\
&= \sum_{p\in V}\sigma(p)\cdot\pbra{\sum_{e\in E\colon u_e=p} \mu_e}
\tag{by rearranging the summation}\\ 
&= \sum_{p\in V} \lambda_p \sigma(p)
\tag{by the definition of $\lambda$}\\
&=\pi_1[\lambda]= \pi_1[a+\lambda]-\pi_1[a],
\tag{by the definition of $\pi_1$}
\end{align*}
which passes the test in \Cref{itm:pcpp-linear-4}.
In all, $A$ accepts $\pi_1\circ\pi_2=\pwh(\sigma)\circ \pi_2$ with probability $1$.
\end{proof}

\begin{lemma}[Soundness]\label{lem:pcpp-linear-3}
    Suppose there is a proof $\pi_1 \circ \pi_2$ which $A$ accepts with probability at least $1-\varepsilon$, then there is a solution $\sigma$ to $G$ such that $\Delta(\pi_1,\pwh(\sigma))\le 24\varepsilon$.
\end{lemma}
\begin{proof}

First, note that $\pi_1$ fails at most $4\varepsilon$ fraction of tests in \Cref{itm:pcpp-linear-1}. 
By the soundness of BLR testing~(\Cref{thm:pwh_test}), there exists some $\sigma\in (\Fbb^d)^k$ such that $\Delta(\pi_1,\pwh(\sigma))\le 24\varepsilon$. 
Similarly by \Cref{itm:pcpp-linear-2}, we can find some $\sigma_2\in (\Fbb^d)^{km}$ such that $\Delta(\pi_2,\pwh(\sigma_2))\le 24\varepsilon$. Below, we treat $\sigma$ as a mapping from $V$ to $\Fbb^d$, and $\sigma_2$ as a mapping from $V\times E\to \Fbb^d$.

Now we prove that for every $p\in V, e\in E$, we have $\sigma_2(p,e) = M_e\sigma(p)$. 
Recall the notation from \Cref{itm:pcpp-linear-3}.
For any fixed $\lambda,\mu$ (and thus $\gamma,M_0$ are also fixed), by \Cref{fct:pwh_correct} and a union bound, with probability at least $1-96\varepsilon$ over random $a,b$, both of the following two equations hold
\begin{align}
\pi_2[b+\gamma]-\pi_2[b] &= \sum_{p\in V,e\in E}{\lambda_p\mu_e\sigma_2(p,e)},\label{eq:test31}\\
\pi_1[a+\lambda]-\pi_1[a] &= \sum_{p\in V}\lambda_p\sigma(p).\label{eq:test32} 
\end{align}
Recall the definition of $M_0$.
By taking a difference between the LHS and RHS of~\Cref{eq:test3} and plugging in \Cref{eq:test31} and \Cref{eq:test32}, we can deduce that, with probability at least $1-96\varepsilon$,
$$
(\pi_2[b+\gamma]-\pi_2[b])-M_0(\pi_1[a+\lambda]-\pi_1[a]) = \sum_{p\in V,e\in E}{\lambda_p\mu_e\left(\sigma_2(p,e)-M_e\sigma(p)\right)}
=\lambda^\top(M_1-M_2)\mu,
$$
where we define matrix $M_1$ by setting the $(p,e)$-th entry as $\sigma_2(p,e)$ and matrix $M_2$ by setting the $(p,e)$-th entry as $M_e\sigma(p)$.

If $\sigma_2(p,e)\ne M_e\sigma(p)$ for some $p\in V, e\in E$, then $M_1\neq M_2$.
By \Cref{lem:rsp} with $\ell=|V|$ and $\ell'=|E|$, we have $\Pr_\lambda\sbra{\lambda^\top M_1\neq\lambda^\top M_2}\ge1-\frac1{|\Fbb|}$.
Then by another round of \Cref{lem:rsp} with $\ell=|E|$ and $\ell'=1$, we have $\Pr_{\lambda,\mu}\sbra{\lambda^\top M_1\mu\neq\lambda^\top M_2\mu}\ge\pbra{1-\frac1{|\Fbb|}}^2\ge\frac14$ as $|\Fbb|\ge2$.
By a union bound, for at least $\frac14-96\varepsilon$ fraction of the tests in \Cref{itm:pcpp-linear-3}, both \Cref{eq:test31} and\Cref{eq:test32} hold, yet the difference above is non-zero. 
This means at least $\frac14-96\varepsilon$ fraction of tests in \Cref{itm:pcpp-linear-3} are violated.
On the other hand, since $A$ accepts $\pi_1\circ\pi_2$ with probability at least $1-\eps$, at most $4\eps$ fraction of the tests in \Cref{itm:pcpp-linear-3} can be violated.
This gives a contradiction as $\varepsilon<\frac1{400}$.

Finally, we prove that $\sigma$ is a solution of $G$. We focus on tests in \Cref{itm:pcpp-linear-4} and recall the notation there.
By a union bound, with probability at least $1-96\varepsilon$ over random $a,b$ for any fixed $\mu$ (and thus $\lambda,\gamma$ are also fixed), both equations below hold:
\begin{align}
\pi_2[b+\gamma]-\pi_2[b] &= \sum_{e\in E}{\mu_e\sigma_2(v_e,e)} = \sum_{e\in E}{\mu_eM_e\sigma(v_e)},\label{eq:test41}\\
\pi_1[a+\lambda]-\pi_1[a] &= \sum_{e\in E}\mu_e\sigma(u_e),\label{eq:test42} 
\end{align}
where the second equality in \Cref{eq:test41} follows from our analysis for \Cref{itm:pcpp-linear-3} above. 
If $\sigma(u_e)\ne M_e\sigma(v_e)$ for some $e\in E$, then by \Cref{lem:rsp}, we have 
$$
\Pr_\mu\sbra{\pi_2[b+\gamma]-\pi_2[b]\neq\pi_1[a+\lambda]-\pi_1[a]}\ge1-\frac1{|\Fbb|}\ge\frac12.
$$
By a union bound, for at least $\frac12-96\varepsilon$ fraction of tests in \Cref{itm:pcpp-linear-4}, both \Cref{eq:test41} and \Cref{eq:test42} hold, yet their difference is non-zero. 
This means $\pi_1\circ \pi_2$ violates at least $\frac12-96\varepsilon$ fraction of tests in \Cref{itm:pcpp-linear-4}, contradicting the fact that $\pi_1\circ \pi_2$ can only violate at most $4\varepsilon$ fraction of tests, recalling again our assumption that $\varepsilon<\frac1{400}$. Thus $\sigma$ is indeed a solution of $G$, completing the proof.
\end{proof}

\Cref{prop:pcpp-linear-constraint} immediately follows from the combination of \Cref{lem:pcpp-linear-1}, \Cref{lem:pcpp-linear-2} and \Cref{lem:pcpp-linear-3}.

\bibliographystyle{alpha}
\bibliography{ref}

\end{document}